\documentclass[aps,superscriptaddress,twocolumn, nofootinbib, longbibliography, pra, 10pt]{revtex4-1}

\usepackage{lipsum, amsmath, amsthm, amsfonts, physics, adjustbox, bm, bbm, graphicx, dsfont, wrapfig, enumerate}
\usepackage{xcolor, textpos}
\usepackage[hidelinks]{hyperref}
\usepackage{thmtools}
\usepackage{thm-restate}
\usepackage{cleveref}
\usepackage{appendix}
\usepackage{comment}
\usepackage{tikz}
\usepackage{xr} 
\usepackage{xurl}
\usepackage{subcaption}

\newcommand{\tuwien}{Institute of Computer Engineering, TU Wien, Austria}
\newcommand{\ethz}{Institute for Theoretical Physics, ETH Zürich, Switzerland}
\renewcommand{\vec}[1]{\boldsymbol{#1}}

\newcommand{\E}[2]{\ensuremath{\mathbb E_{#1}\left[ #2 \right]}}
\newcommand{\Var}[1]{\ensuremath{\operatorname{Var}\left[ #1 \right]}}
\renewcommand{\exp}[1]{\ensuremath{\operatorname{exp}\left( #1 \right)}}

\newcommand{\mmd}[1]{\ensuremath{\operatorname{MMD}^2\left( #1\right)}}
\newcommand{\poly}[1]{\ensuremath{\operatorname{poly}\left( #1\right)}}

\newcommand{\prob}[2][]{\ensuremath{\operatorname{Prob}_{#1}\left( #2 \right)}}

\newcommand{\estMMD}{\ensuremath{\widehat{\operatorname{MMD}}^2}}
\newcommand{\MMD}{\ensuremath{{\operatorname{MMD}}^2}}

\newcommand{\Dir}[1][\vec\alpha]{\ensuremath{\operatorname{Dir}\left( #1 \right)}}
	
\newtheorem{definition}{Definition}
\newtheorem{lemma}{Lemma}
\newtheorem{theorem}{Theorem}

\newtheorem{corollary}{Corollary}

\let\section\section

\let\section\section	

\begin{document}

\title{Limits of quantum generative models with classical sampling hardness}
\author{Sabrina Herbst}
\email{sabrina.herbst@tuwien.ac.at}
\affiliation{\tuwien}
\author{Ivona Brandi\'c}
\affiliation{\tuwien}
\author{Adrián Pérez-Salinas}
\affiliation{\ethz}

\begin{abstract}
Sampling tasks have been successful in establishing quantum advantages both in theory and experiments. 
This has fueled the use of quantum computers for generative modeling to create samples following the probability distribution underlying a given dataset. 
In particular, the potential to build generative models on classically hard distributions would immediately preclude classical simulability, due to theoretical separations.
In this work, we study quantum generative models from the perspective of output distributions, showing that models that anticoncentrate are not trainable on average, including those exhibiting quantum advantage. 
In contrast, models outputting data from sparse distributions can be trained.
We consider special cases to enhance trainability, and observe that this opens the path for classical algorithms for surrogate sampling. 
This observed trade-off is linked to verification of quantum processes.
We conclude that quantum advantage can still be found in generative models, although its source must be distinct from anticoncentration. 
\end{abstract}
\maketitle

\section{Introduction}
Quantum computers are known to have extended capabilities with respect to their classical counterparts, which has been established through examples such as Shor factorization~\cite{Shor_1997}, Grover search~\cite{grover}, or experimental demonstrations of quantum advantages~\cite{Arute_2019, Broome_2013}. 
Inspired by the success of classical Machine Learning (ML), 
quantum machine learning (QML) explores recipes to process data through quantum resources and surpass current computational capabilities. 
The presence of data, however, incorporates subtleties for showing formal separations, since simulating a computational process and learning from data are fundamentally different tasks ~\cite{schreiber2023classical, gil-fuster2024relation,huang2021power}.
Still, in certain cases, it is possible to show robust separations~\cite{jerbi2023shadows, gyurik2024quantum, Sweke2021quantumversus, born-supremacy}, e.g., by planting the seed of a quantumly easy and classically hard problem inside a learning algorithm~\cite{liu2021rigorous}.  
For that, strong assumptions on data, as the central element of (Q)ML, are necessary. 
However, realistic data might be noisy, incomplete, imperfect or redundant, reducing the scope of applicability of these separations.
The central difficulty thus lies in finding practically useful algorithms that perform better or are more efficient compared to classical counterparts.

Quantum computers suit a natural interpretation as generative models, i.e., ML models learning data distributions to subsequently sample from them.
Sampling is inherent to quantum processes, since measurements are the only way to retrieve information from quantum states. 
In particular, the measurement process probabilistically generates samples based on the underlying state through Born's rule.
This property is exploited in generative QML (GQML) models, such as parameterized quantum circuits, Boltzmann machines or Gibbs states~\cite{barthe-gen, shenkurkin, kurkinshen, zoufal-phd}.
We recall that sampling tasks were the first framework where quantum advantage was found, with the examples of Boson sampling~\cite{aaronson2011computational}, Instantaneous Quantum Polynomial (IQP)~\cite{Bremner_2010, bremner2016average} circuits or random quantum circuits (RQC)~\cite{boixo, Arute_2019}. Using these circuits as subroutines open the potential for quantum-advantageous generative models~\cite{PhysRevA.108.042406, recioarmengol2025iqpoptfastoptimizationinstantaneous}.

In this work, we show that classical hardness of circuits has its reflections in QGML.  
Shortly, classical hardness requires \textit{anticoncentration}, that is, the underlying probability distribution is on average \textit{almost but not exactly} uniform.
As a consequence, most loss functions required to train the generative models will strongly concentrate, preventing any training. 
The result shows reminiscence of barren plateaus \cite{McClean_2018, larocca2024review, cerezo2025does}.
This observation is tightly related to verification of a quantum sampling task, and deeply rooted in the assumptions allowing for quantum advantage.

To show these results, we analytically study three abstractions of quantum circuits; (1) a circuit description corresponding to product states, (2) a family yielding pseudo-independent probability distributions, which we connect to quantumly-advantageous circuits, and (3) pseudo-independent sparse distributions, which are related in spirit to classically verifiable circuits with quantum advantage~\cite{aaronson2024verifiablequantumadvantagepeaked, zhang2025complexityhardnessrandompeaked}.
We show that the first two families lead to exponential concentration of the square distance between probability distributions, although only pseudo-independent states can show classical hardness. 
For the sparse case, loss functions do not concentrate, however, due to missing anticoncentration, classical sampling hardness is not established. 
We further provide numerical experiments to interpolate between analytical results and practically relevant circuits.

The concentration results we find have direct implications for GQML, targeting the subtle difference between sampling hardness and utility of the model. 
The lack of trainability, implied by the on-average impossibility to distinguish between the target function and the produced one, preclude the on-average usage of hard-to-sample generative models. 
While it is possible to find trainable instances, the sampling hardness guarantees are severely compromised in this case. 

The paper is organized as follows. 
\Cref{sec.background} provides relevant background on QGML and classical sampling hardness. 
\Cref{sec:setup-loss-function} states our theoretical framework and proves concentration for loss functions. 
\Cref{sec:numerical-exploration} numerically extends the theoretical statements. 
Connection to verification and implications for QML are detailed in \Cref{sec.verification} and \Cref{sec.implications_qml}. 
We conclude in \Cref{sec.conclusion}.

\section{Background}\label{sec.background}

\subsection{Quantum computers as generative models}

Quantum computers are easily interpreted as discrete sampling machines, in which the probability of sampling a possible outcome, e.g., a bitstring, follows Born's rule. Given a circuit $C$, this probability is given as 
\begin{equation}
    p_C(x) = \left\vert \bra x C \ket 0 \right\vert^2, \quad x \in \{0, 1\}^n.
\end{equation}
Interestingly, the task of sampling bitstrings from a quantum computer has been fertile in showing quantum advantages, both theoretically and experimentally, through, e.g., IQP~\cite{bremner2016average}, RQC~\cite{boixo, aaronson2016complexitytheoretic, Arute_2019} or Boson sampling~\cite{aaronson2011computational}. 

This perspective opens a path for using quantum computers as a generative ML model, where the goal is 
to learn the underlying distribution $q$ of a dataset of samples. 
For GQML, this is done, e.g., with Quantum Circuit Born machines~\cite{Benedetti_2019}, Quantum Boltzmann Machines~\cite{PhysRevX.8.021050} and Quantum Generative Adversarial Networks~\cite{PhysRevLett.121.040502}.
Any generative model based on a sampling task with provable quantum advantage has the potential to inherit the classical hardness, thus, no classical algorithm could - even approximately - substitute the quantum machine. 
This has, in particular, been explored for IQP circuits~\cite{PhysRevA.108.042406, recioarmengol2025trainclassicaldeployquantum}, which even allow for universality if ancillary qubits are used~\cite{kurkinshen}. 
In this scenario, classical sampling hardness -- in the assumptions of the original IQP hardness proofs -- is, to the best of our knowledge, an open problem. 

From a practical point of view, a (Q)GML model must be trainable to be useful, that is, there must exist an efficient algorithm that finds one particular instance of the model that accurately mimics the training data \cite{gil-fuster2024relation}. 
The notion of trainability has been widely addressed in quantum variational models, for instance in the context of barren plateaus (BP), yielding fundamental trade-offs between trainability and classical simulability~\cite{McClean_2018, cerezo2025does}.
In a nutshell, very general models with no proper bias behave essentially randomly and are not trainable due to the exponential-in-qubits concentration of most loss functions over the search space.

For all generative models, training requires an estimation of the difference between the current candidate and the target probability distribution.
We distinguish between explicit and implicit losses depending on whether samples or (estimated) probability vectors are used for calculation. 
In a quantum setting, implicit losses can be encoded into observables.  
If their expectation values can be efficiently estimated classically, 
the training can be delegated on classical surrogates~\cite{rudolph-generative-modelling, van-den-nest}, thus avoiding pitfalls when using current hardware. 

\subsection{Sampling hardness}

\begin{figure}[t!]
    \centering
    \begin{tikzpicture}
    \node at (0, 0) {\includegraphics[width=0.9\linewidth]{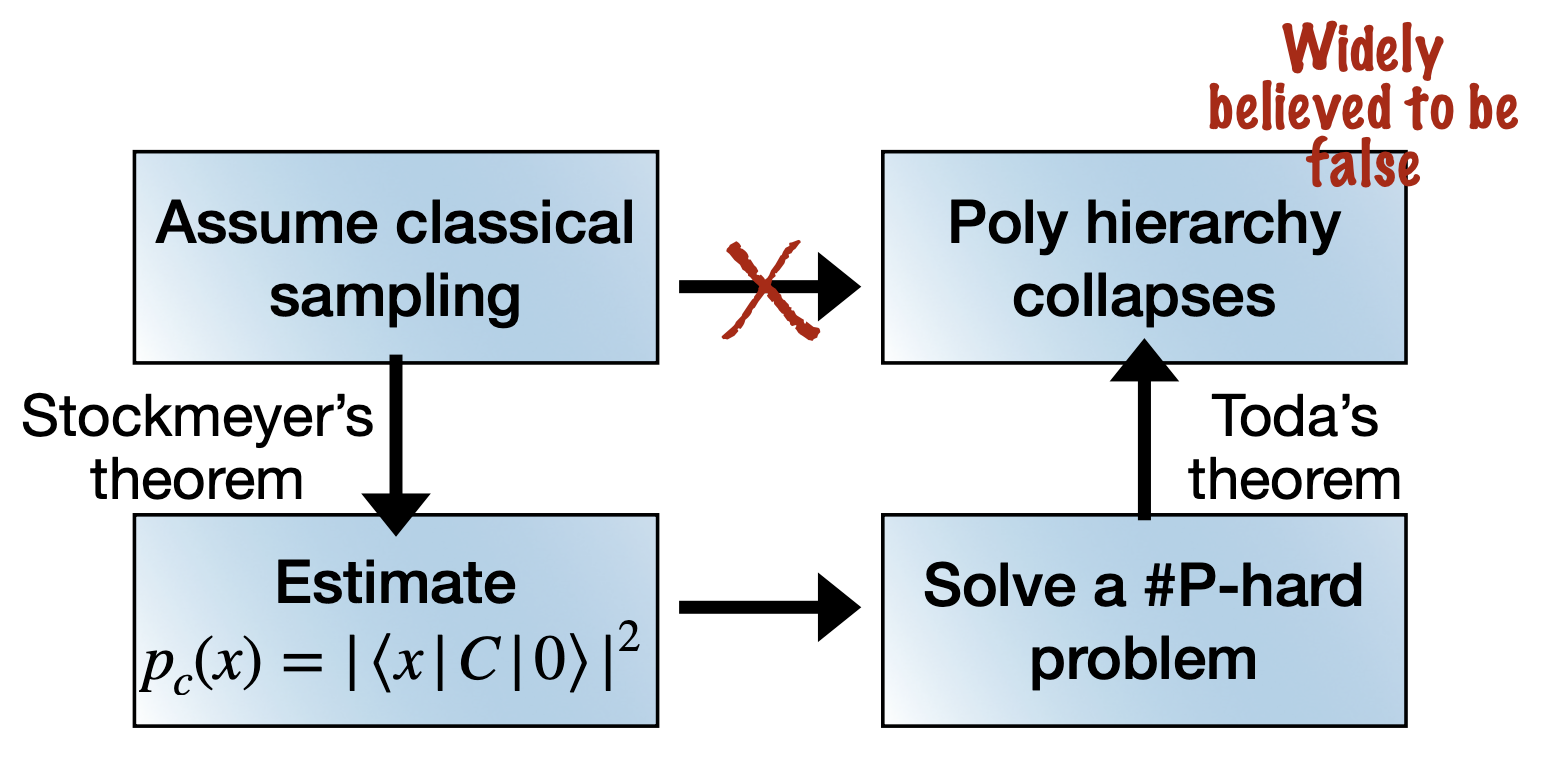}};
    \node (toda) at (1.5, -.3) {\cite{toda}};
    \node (stockmeyer) at (-1.5, -.3) {\cite{stockmeyer}};
    \end{tikzpicture}
    \caption{Schematic description for the proof of classical hardness on sampling tasks. 
    The initial step is assuming the existence of a classical machine outputting samples from the distribution generated by the quantum device (to $\epsilon$ error). 
    The samples are used as an oracle to approximate $p_C(x) = \left\vert \bra{x} C \ket 0\right\vert^2$, which encodes a \#P-hard problem. 
    Multiplicative estimation of this quantity would imply the collapse of the polynomial hierarchy, which is conjectured to be false. Hence, classical sampling cannot be possible. }
    \label{fig.scheme-hardness}
\end{figure}

As mentioned before, sampling has proven to be a fruitful ground for quantum advantage~\cite{bremner2016average,boixo, aaronson2016complexitytheoretic, Arute_2019,aaronson2011computational}. The hardness proofs all follow a common pattern, rooted in hard problems from complexity theory, which is schematically depicted in \Cref{fig.scheme-hardness}.
In particular, to show hardness, the quantum circuit must encode the solution of a \#P-hard problem into the probability of sampling a given bitstring.
By Stockmeyer's theorem~\cite{stockmeyer}, sampling from the circuit then allows estimating said probability.
The existence of a classical machine that samples from a probability distribution $\epsilon$-close to the one generated by the quantum machine, thus, implies a multiplicative estimation of the solution of a \#P-problem, and more importantly the collapse of the polynomial hierarchy via Toda's theorem~\cite{toda}. This last event is widely believed to be false under common assumptions in complexity theory, thus proving (under the same assumptions) classical hardness\footnote{Several subtleties in these proves are hidden due to space and narrative constraints. We refer the interested reader to the original sources for an in-depth discussion.}.

Crucially, the \#P-hard requirement previously considered is only compatible with the estimation of $p_C(x)$ to multiplicative precision, because the output of a quantum circuit is a number between $0$ and $1$, and not an integer number of solutions, as in the conventional description of a \#P problem. 

Further, sampling hardness in the considered examples is an average statement established through anticoncentration, usually presented in the literature in two different albeit equivalent forms~\cite{dalzell2022random}, namely
\begin{align}
    \prob{p_c(x) \geq \frac{y}{2^n}} & \geq \beta & \forall x; \;  y, \beta > 0 \label{eq.anticoncentration1}, \\
    2^{2n}\E{}{p_c(x)^2} & \geq \beta^\prime & \forall x; \; \beta^\prime > 1\label{eq.anticoncentration2}.
\end{align}
A paradigmatic example of a distribution of this kind is the Porter-Thomas distribution \cite{porter1956fluctuations}. 
Intuitively speaking, anticoncentration means that the typical probability distribution of a quantum model is \textit{almost but not exactly} uniformly distributed. 
Hence, with high probability, $p_C(x) \in \Theta(2^{-n})$. 
Under this assumption, additive precision, 
\begin{equation}
    \left\vert p_C(x) - \hat{p}_C(x)\right\vert<\epsilon, \quad \epsilon^{-1} \in \mathcal O\left(\poly n\right)
\end{equation}
becomes meaningless. 
In fact, the diagonal structure of IQP circuits allows for the existence of an efficient classical algorithm for estimating $p_C(x)$ to $1/\poly{n}$ additive precision~\cite{van-den-nest}. 
However, this estimation of $p_C(x)$ does not solve the hard problem lying at the core of quantum advantage.

\section{Concentration of loss functions}\label{sec:setup-loss-function}

In this section, we model three relevant examples of QGML architectures with three different families of probability distributions and analyze statistical properties of their loss.
More specifically, we consider i) product distributions, representing product quantum states, ii) pseudo-independent probability distributions, representing and drawing inspiration from RQC and IQP, and iii) sparse pseudo-independent distributions, representing an adapted version of peaked quantum circuits~\cite{aaronson2024verifiablequantumadvantagepeaked, zhang2025complexityhardnessrandompeaked}.

In summary, our results, shown in \Cref{tab:summary}, show that any loss on anticoncentrated distributions will exponentially concentrate. 
This implies that quantum models with sampling hardness rooted in anticoncentration will not be trainable. 
In addition, the values $p_C(x)$ will behave approximately as independent variables, requiring exponential precision to observe correlations, thus, hindering the applicability of such circuits.
For these reasons, such architectures will not provide a usable generative model on average.

\begin{table*}
    \begin{tabular}{||c||c|c|c||}\hline
        Distribution & Product, def. \ref{def.product_distribution} & Pseudo-independent, def. \ref{def.independent_distribution} & Peaked pseudo-independent, def. \ref{def.peaked_distribution} \\ \hline
        \parbox[b]{3cm}{Average behavior: \vspace{3mm} \\ Black lines are $\mu - \sigma$ \\ Blue lines are $\mu + \sigma$\vspace{10mm}} & \includegraphics[width = 4cm]{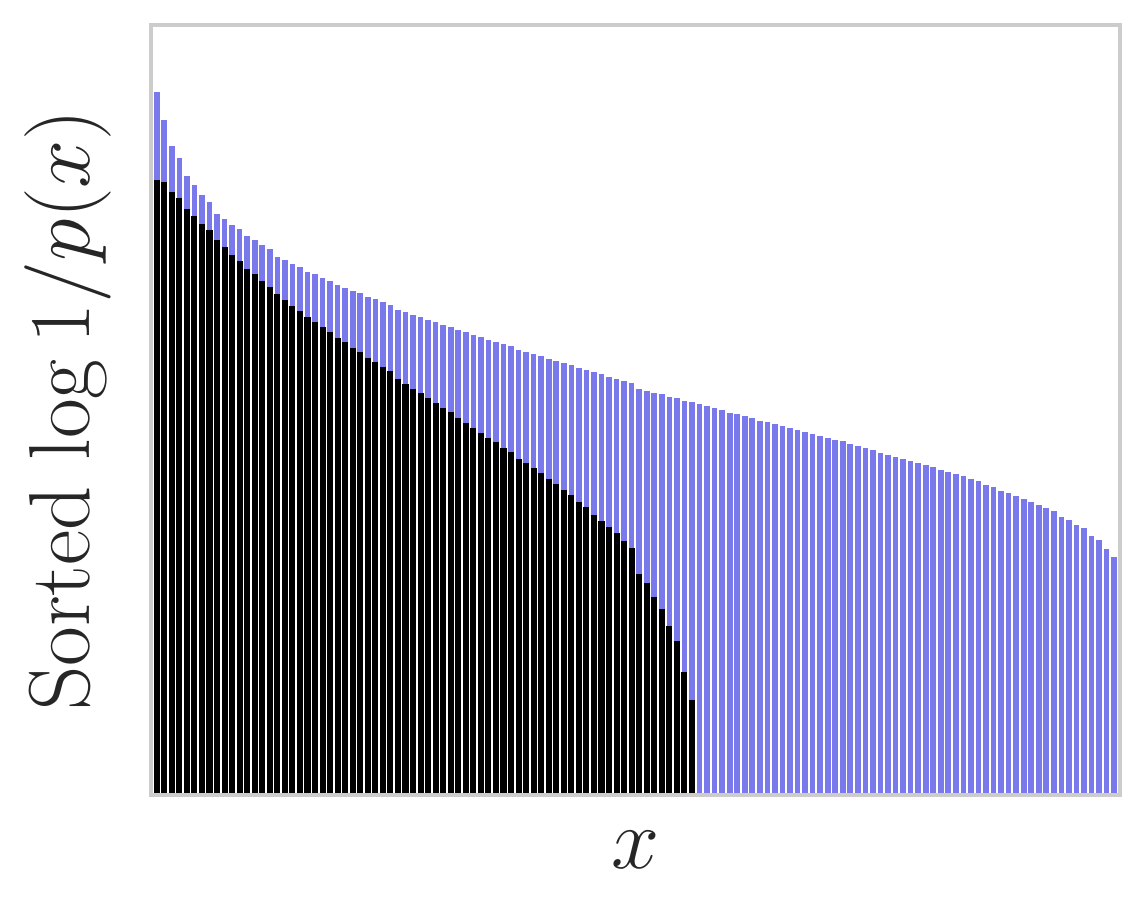} & \includegraphics[width = 4cm]{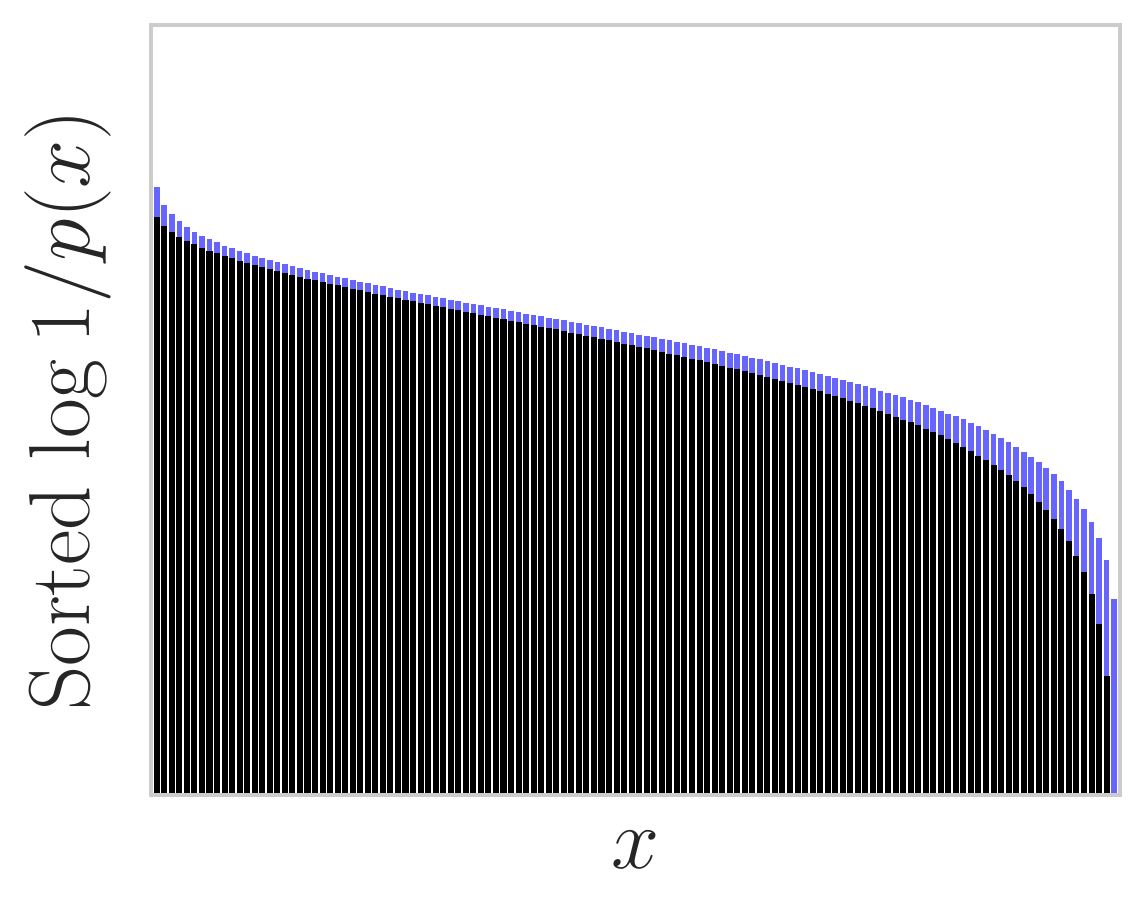} & \includegraphics[width = 4cm]{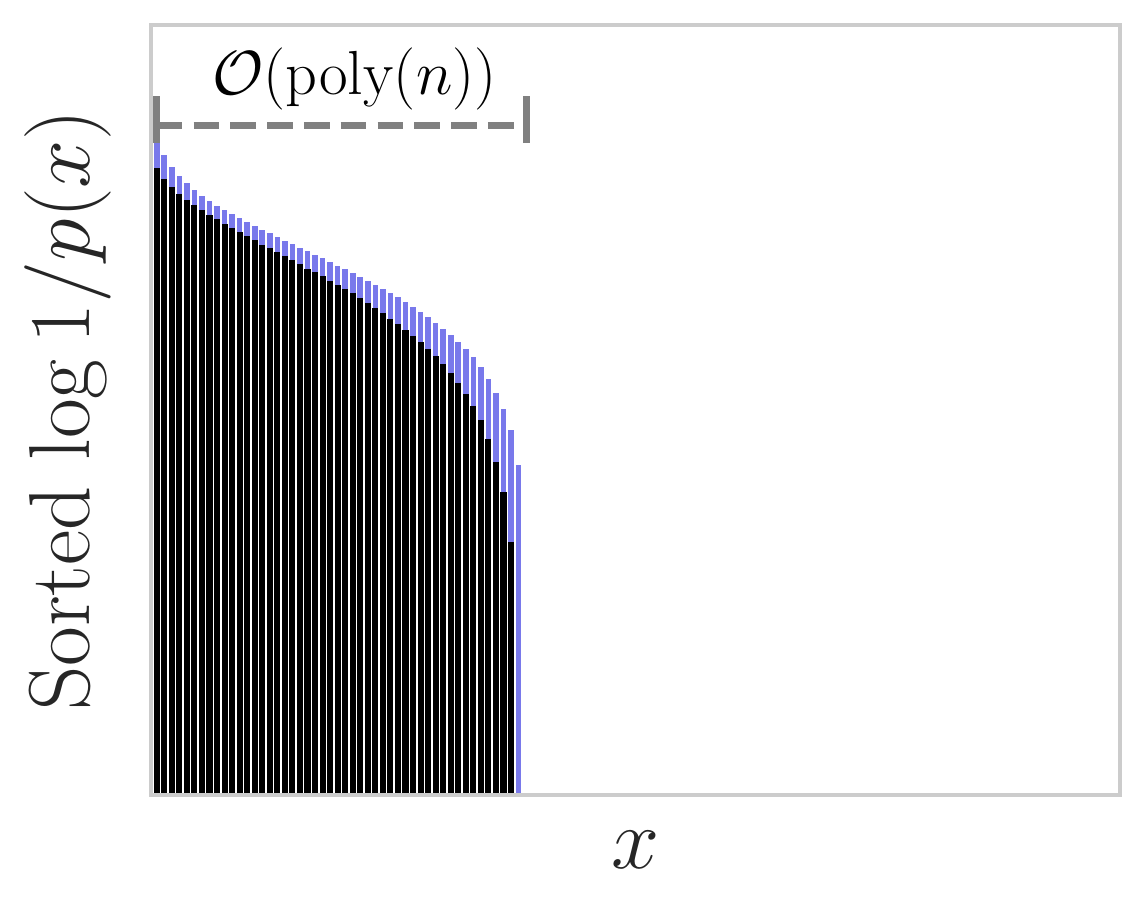} \\\hline
        Quantum Circuit & Product circuits & IQP~\cite{Bremner_2010}, RQC \cite{boixo} & Peaked circuits \cite{aaronson2024verifiablequantumadvantagepeaked} \\ \hline
        Anticoncentration & $\cross$,  prop. \ref{prop.concentrationproduct} & $\checkmark$, prop. \ref{prop.anticoncentrationindependent} & $ \cross$, prop. \ref{prop.concentrationpeaked} \\ \hline
        Loss concentration & $\checkmark$, prop. \ref{prop.msdproduct} & $\checkmark$, prop. \ref{prop.msdpseudo} & $\cross$, prop. \ref{prop.msdpeaked} \\ \hline 
        Classical sampling & \checkmark \cite{bravyi2021classical} & $\cross$  \cite{bremner2016average, boixo} & ?, see sec. \ref{sec.peaked_circuits} \\ \hline
        Clas. estimation $p_C(x)$ & \checkmark \cite{cirac2021matrix} & $\cross$ \cite{bremner2016average, boixo} & $\cross $ \cite{zhang2025complexityhardnessrandompeaked}\\ \hline
    \end{tabular}
    \caption{Properties of the three considered distributions in this work. 
    A typical example of each distribution lies within the blue area. 
    $\mu$ and $\sigma$ are element-wise mean and standard deviations of the sorted probability vectors $p_C(x)$. 
    Anticoncentration is the standard definition from \Cref{eq.anticoncentration1}. 
    Classical sampling and estimation of $p_C(x)$ correspond to the existence of efficient classical algorithms to perform these tasks.}
    \label{tab:summary}
\end{table*}

Our framework assumes that the target distribution is included in the family of distributions spanned by the generative model.
In fact, we assume that $q$ is a typical outcome of the generative model, mirroring the average-case hardness of the distribution families.
However, our results hold even in the case where $q$ is an outlier, since it will be approximately equally distant from all typical candidates $p$, see \Cref{sec.implications_qml}.

\subsection{Three analytical families of distributions}

We now introduce the different types of distributions covering the three extreme cases of relevance to generative QML. For each probability distribution, we study their (anti)concentration and the typical distance of two probability distributions of each kind, in \Cref{sec.preliminaries_loss_functions}. The main distance we consider is the squared distance (SD), see \Cref{def:msd}.

\subsubsection{Product distributions}\label{sec.def_product}

The first distribution we consider is the \textit{product distribution}. 
To generate an instance of a product distribution, we sample $n$ uniformly-random real numbers $\vec a, a_i \in [0, 1]$. 
Then, bitstrings are sampled bit by bit; $0$ with probability $a_i$, $1$ with probability $1 - a_i$ as a collection of Bernoulli processes.
This procedure is a direct recipe for an efficient classical algorithm sampling from this probability distribution, and it also returns the probability of sampling each bitstring $p_{\vec a}(x)$.
We formalize the distribution family as follows.

\begin{figure}
    \centering
    \includegraphics[width=.9\linewidth]{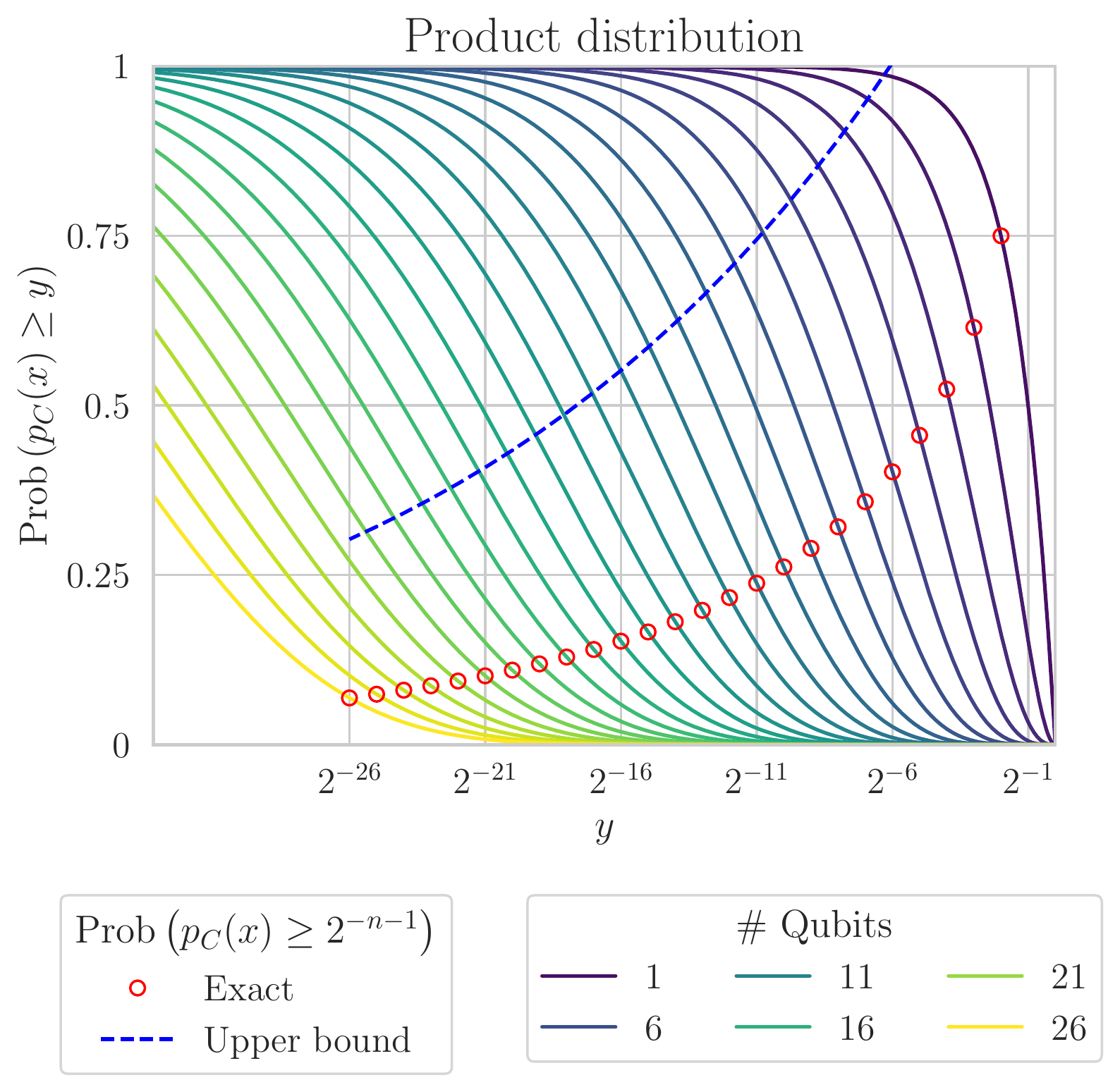}
    \caption{Concentration of the product distribution. The solid lines depict the quantity $\prob{p_C(x) \geq y}$, for different numbers of qubits $n \in [1, 26]$, the red circles highlight the values at $y = 2^{-(n + 1)}$, and the blue line is the upper bound from \Cref{prop.concentrationproduct}. The relevant observation for concentration is that these red circles render exponentially vanishing values as $n$ increases.}
    \label{fig.product_concentration}
\end{figure}

\begin{definition}[Product distribution]\label{def.product_distribution}

Let $\vec a \in [0, 1]^n$ be uniformly distributed on its domain. A bitstring $x \in \{ 0, 1\}^n$ is randomly distributed according to the product distribution if 
\begin{equation}
    p_{\vec a}(x) = \prod_{i = 1}^n \left(x_i + (-1)^{x_i}a_i\right).
\end{equation}

\end{definition}

The family of quantum circuits that generates this product distribution is just all quantum circuits composed by a layer of single-qubit unitaries, sampled from the single-qubit Haar-random distribution~\cite{diestel2014joys}. We use this example as the representative of all quantum circuits with low entanglement.
An immediate extension are matrix product states (MPS)~\cite{cirac2021matrix}, which fall back to product states for MPS with bond dimension one.
MPS are classically simulable, both for estimating $p_C(x)$ directly and sampling~\cite{bravyi2021classical}.

As a first step, we consider (anti)concentration properties of the product distribution.
Observe that $p_C(x)$ enjoys strong symmetry over $x$ due to the uniform distribution over the parameter $\vec a$. In particular, a change in one bitstring is equivalent to the map $a_i \mapsto 1 - a_i$. Consider, without loss of generality, $a_i \leq 1/2$. Then, those bitstrings with more $0$'s are geometrically more likely to appear than bitstrings with more $1$'s. 
This implies that the probability distribution will be biased towards a few bitstrings, which will become more prevalent with larger $n$.
Formally, this phenomenon reads as follows. 
\begin{restatable}[Concentration of product distribution]{proposition}{concentrationproduct}\label{prop.concentrationproduct}
Let $p_{\vec a}$ be a product distribution. This probability distribution concentrates, as
\begin{equation}
    \forall x: \quad \prob{p_{\vec a}(x) \geq \frac{y}{2^n}} \leq \frac{\exp{-\mathcal O(n)}}{\sqrt y}.
    \end{equation}
\end{restatable}
The proof relies on finding the volume of the values of $\vec a$ such that $\prod _ia_i = y'\rightarrow \sum_{i} \log a_i = \log y'$. After a change of variables, this is equivalent to computing the volume of a simplex, which scales as $\mathcal O\left((y'/n)^n\right)$. The integral for all compatible simplices and an upper Chernoff bound suffices to simplify the expression to the obtained result. A graphical description of the obtained formulae is available in \Cref{fig.product_concentration}, and the proof can be found in \Cref{app:product-distribution}.

\subsubsection{Pseudo-independent distributions}\label{sec.def-pseudo-independent}
The second family of probability distribution is what we call \textit{pseudo-independent distributions}. These arise from i.i.d. sampling positive random variable following some distribution, and normalizing by the sum. $p_C(x)$ then corresponds to the normalized random variables.

\begin{definition}[Pseudo-independent distributions]\label{def.independent_distribution}
Let $Y_1, Y_2, \ldots Y_N \in [0, \infty)$, be a set of $N = 2^n$ independent and identically distributed random variables. A bitstring $x \in \{ 0, 1\}^n$ is randomly distributed according to a pseudo-independent distribution if 
\begin{equation}
    p(x) = \frac{Y_x}{\sum_{j = 1}^N Y_j}. 
\end{equation}
By definition, $p(x) \in [0, 1]$, $\sum_{x = 1}^N p(x) = 1$, and the marginal distributions are identical for all $x$.
\end{definition}

We call this family pseudo-independent because all possible correlations need to come from the normalization, and these correlations decrease as $N$ increases. In a nutshell, the average of all i.i.d. sampled variables needs to concentrate around the true mean of the underlying distribution, following the central limit theorem. Formally:   
\begin{corollary}\label{cor.approx_indep}
The random variables from \Cref{def.independent_distribution} are approximately independent, as
\begin{equation}
    \left\vert p(x) - \frac{Y_x}{\E{}{Y} N} \right\vert \leq \frac{Y_x}{\E{}{Y} N} \frac{\Var{Y}}{N k^2} 
\end{equation}
with probability at least $1- k^{-2}$, for $Y_x$ being independent variables. 
\end{corollary}
For a more detailed statement, we refer the reader to \Cref{le.diff_average_Y}, in \Cref{app:independent-distribution}. 
Notice that the pseudo-independence assumption will only hold if the variance of the underlying distribution is defined. One could in principle force this variance to increase with the number of dimensions to break this assumption. This is considered in \Cref{sec:numerical-exploration} under the form of \textit{Pareto distributions}. 

The pseudo-independent distributions are chosen as the representative of all quantum circuits whose output are \textit{almost but not exactly} uniformly distributed, such as IQP and RQC. 
In fact, these circuits approximately follow this definition based the original works on quantum advantage. 
It is commonly known that RQC follow a Porter-Thomas distribution~\cite{boixo}, that is
\begin{equation}
    \prob{p_C(x) \geq y} \approx \exp{-N y}, 
\end{equation}
which obviously anticoncentrates. 
A similar behavior can be observed for IQP circuits~\cite{boixo}. 
In more detail, the Porter-Thomas distribution is the marginal of a Dirichlet distribution~\cite{multivariate}, i.e., a Beta distribution~\cite{bailey1992distributional}. 
\Cref{def.independent_distribution} is then matched if the underlying variables are i.i.d. samples from a Gamma distribution. 
We visualize the connections in \Cref{fig:distribution-connections} and provide a detailed explanation in \Cref{app:pconnection-dirichlet-indepdent}.

\begin{figure}
    \centering
    \begin{tikzpicture}
    \node (scheme) at (0, 0) {\includegraphics[width = .8\linewidth]{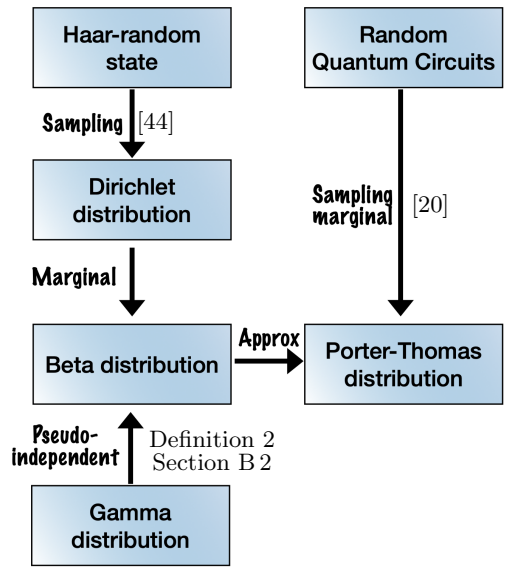}};
    \node (citeboixo) at (2.3, 1.00) {\cite{boixo}};
    \node (citeadrian) at (-1.35, 2.1) {\cite{bonetmonroig2024averagerandomnessverificationsets}};
    \node (appendix) at (-.6, -2.1) {\Cref{def.independent_distribution}};
    \node (appendix) at (-.6, -2.4) {\Cref{app:pconnection-dirichlet-indepdent}};
    
    \end{tikzpicture}
    \caption{Connection between sampling tasks with quantum advantage and pseudo-independent distributions. Quantum models with classical sampling hardness have approximated analytical descriptions matching our pseudo-independent probability distributions with underlying Gamma distribution. IQP circuits exhibit similar behavior.}
    \label{fig:distribution-connections}
\end{figure}

Very general assumptions suffice to establish anticoncentration for pseudo-independent distributions. Since the correlations appear only through normalization, each $p_C(x)$ approximately follows the underlying probability distribution.  
\begin{restatable}[Anticoncentration of pseudo-independent distribution]{proposition}{anticoncentrationindependent}\label{prop.anticoncentrationindependent}
Let $p$ be an instance of the pseudo-independent distributions. This probability distribution anticoncentrates, as
\begin{equation}
    \forall x: \quad \prob{p(x) \geq \frac{y}{2^n}} \in \Omega(1).
    \end{equation}
\end{restatable}
A detailed statement can be found in \Cref{app:independent-distribution}, and a graphical description in \Cref{fig.pseudo_concentration}. Notice that this statement captures the dominant term and not the correlations due to its small contribution, following \Cref{cor.approx_indep}.

\begin{figure}
    \centering
    \includegraphics[width=.9\linewidth]{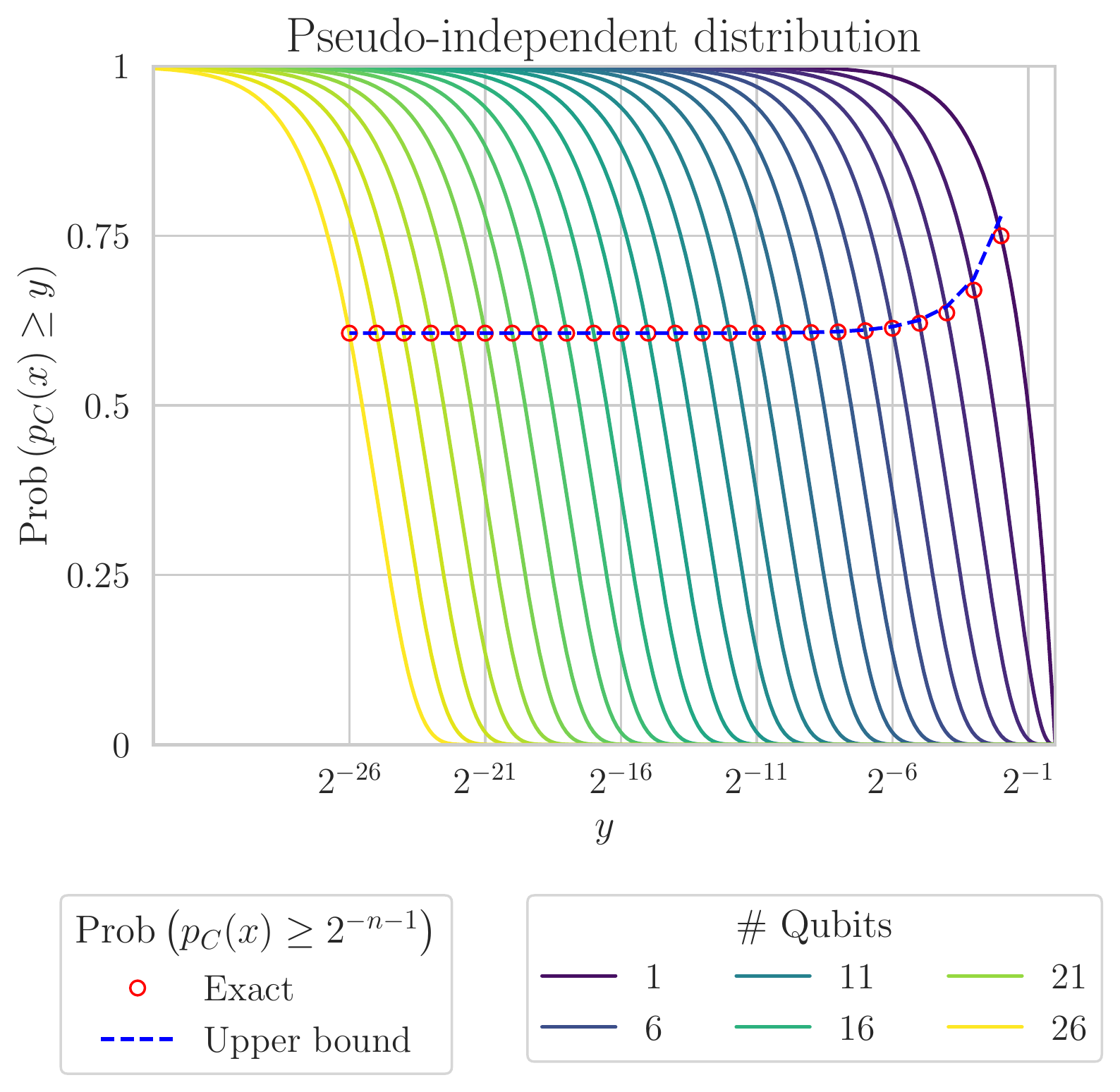}
    \caption{Concentration of the pseudo-independent distribution for random states, i.e., those following a Beta distribution. The solid lines depict the quantity $\prob{p_C(x) \geq y}$, for different numbers of qubits $n \in [1, 26]$, the red circles highlight the values at $y = 2^{-(n + 1)}$, and the blue line is the approximation from \Cref{prop.anticoncentrationindependent}. The relevant observation for concentration is that these red circles appear at constant values of $y$ independently of $n$.}
    \label{fig.pseudo_concentration}
\end{figure}

\subsubsection{Pseudo-independent peaked distributions}
Last, we consider \textit{pseudo-independent peaked distributions}. Those are similar to pseudo-independent distributions, but restricted to have sparse outcomes, where most of the bitstrings can never be sampled. We define them as follows.

\begin{definition}[Pseudo-independent peaked distributions]\label{def.peaked_distribution}
Let $\bar p(j)$ be the underlying probability distribution of a set of outcomes, $j \in \{0, 1, \ldots, K\}$ as given in \Cref{def.independent_distribution}. A bitstring $x \in \{ 0, 1\}^n$ is randomly distributed according to the peaked pseudo-independent distribution if 
\begin{equation}
    p(x) = \left\{\begin{matrix}
        \bar p(j) & j \in S\\
        0 & {\rm otherwise}
    \end{matrix} \right. ,
\end{equation}
for $S$, a random subset of $\{0, 1, \ldots, N \}$ of cardinality $K$, $K \leq N$.
\end{definition}

The choice of peaked pseudo-independent distribution is inspired by peaked quantum circuits~\cite{aaronson2024verifiablequantumadvantagepeaked}. 
In a conventional definition, there must exist one $x$ so that $p_C(x) \in \Omega(\operatorname{poly}^{-1} (n))$. 
We adapted this approach to have circuits whose outcomes concentrate in a polynomial number of bitstrings.  
Peaked circuits were proposed as novel candidates for quantum advantages that are classically verifiable~\cite{aaronson2024verifiablequantumadvantagepeaked}. In contrast, models with anticoncentration cannot be verified efficiently~\cite{hangleiter}.  
Classical estimation of the probabilities $p_C(x)$ of these circuits is not possible under common assumptions in complexity theory~\cite{zhang2025complexityhardnessrandompeaked, gharibyan2025heuristic}. 

Further, it is immediate that these probability distributions concentrate, which we formalize in the following.
\begin{restatable}[Concentration of peaked pseudo-independent distribution]{proposition}{concentrationpeaked}\label{prop.concentrationpeaked}
Let $p$ be a peaked pseudo-independent distribution. This probability distribution concentrates, as
\begin{equation}
    \forall x: \quad \prob{p_{\vec a}(x) \geq \frac{y}{2^n}} \leq \frac{K}{N} \in \mathcal O \left( \frac{\operatorname{poly}(n)}{2^n}\right).
    \end{equation}
\end{restatable}
The proof is a counting argument. Even in the assumption that $p_C(x)=1$ for some $C$ and $x$, the fraction of instances contributing to the bitstring $x$ is only $K / N$. 

\subsection{Loss functions}\label{sec.preliminaries_loss_functions}
In this section we explore the implications that the three aforementioned families of probability distributions have on loss functions. 
In any training procedure, computing a loss is a necessary step, and, as a bare minimum, it is necessary to compare the loss functions of two configurations of the generative model, to decide which one is better.
To this end, the loss function needs to be estimated to a reasonable level of accuracy.
We show that the most commonly used loss functions exhibit strong concentration properties, and thus no efficient and informative estimation is possible. 

The results in this section are built on top of the square distance (SD) between two probability distributions, i.e., the candidate and target.
\begin{definition}[Squared distance (SD) over probability distributions]\label{def:msd}
    Let $p, q$ be probability distributions over the domain $\mathcal X = \{ 0, 1\}^n$. The square distance of the probability distributions is given by 
    \begin{equation}
        \Delta(p, q) = \sum_{x \in \mathcal X}\left(p(x) - q(x) \right)^2
    \end{equation}
\end{definition}
There exists a fundamental motivation to focus on this distance. 
First, the expectation value of this distance is related to the second statistical moment of the probability distribution, which immediately connects to the definition of anticoncentration from \Cref{eq.anticoncentration2}. 
Hence, the SD contains a direct link to classical hardness of sampling. 
Second, the SD is a simple measure for the similarity of two probability distributions on average, and it is a simple, yet effective, tool to analyze states with randomness, such as the ones in RQC or IQP, lying also at the core of the considered proofs for hardness. 

While useful for generic comparisons between probability distributions, the SD is not sensitive to details in differences that could provide more insightful information for a practitioner. 
To this end, other distances focusing on certain regions of the probability distribution have been proposed, e.g., by utilizing a kernel $k(x, y) \geq 0$ that weights pairs of data points for a more tailored loss function. 
The Maximum Mean Discrepancy (MMD)~\cite{JMLR:v13:gretton12a} distance is a commonly used kernel-based loss in the classical literature, and it has already been employed in the context of GQML~\cite{rudolph-generative-modelling, recioarmengol2025trainclassicaldeployquantum}. More formally, it is defined as follows. 
\begin{definition}[Squared population $\MMD$~\cite{JMLR:v13:gretton12a}]\label{def.mmd}
    Consider independent random variables $x,x'$ and $y,y'$, distributed according to their respective probability distributions $p$ and $q$, and an arbitrary kernel function $k(\cdot,\cdot) \geq 0$. The squared population $\MMD$ is defined as
    \begin{equation}
        \begin{aligned}
            \MMD_k(p,q) = & \\
             \mathbb{E}_{x,x'\sim p}&[k(x,x')] + \mathbb{E}_{y,y'\sim q}[k(y,y')] \\
            & -2\mathbb{E}_{x\sim p,y\sim q}[k(x,y)]
        \end{aligned}
    \end{equation}
\end{definition}
The MMD is a distance on its own for characteristic kernels, \cite{JMLR:v13:gretton12a,NIPS2007_3a077244}, such as the Gaussian or Laplace kernels, thus, vanishes for equal distributions, i.e.,
\begin{equation}
    \MMD_k(p, q) = 0 \Leftrightarrow p = q \quad \textrm{if } k \textrm{ is characteristic}.
\end{equation}

In the case of generative models defined over the domain of bitstrings, previous literature \cite{rudolph-generative-modelling, recioarmengol2025trainclassicaldeployquantum} has employed the Gaussian kernel shown as follows. Here, the 2-norm reduces to a simple Hamming distance, thus, counting distinct bits between $x$ and $y$.
\begin{equation}
    k_\varsigma(x, y) = \exp{-\frac{d_H(x, y)}{2 \varsigma^2}}, 
\end{equation}
Especially for large $\varsigma$, this leads to low-order correlations dominating the loss, which becomes apparent when transforming the $MMD^2$ into the alternative form, given via the Fourier character of $p$ and $q$.
\begin{multline}\label{eq.fourier_mmd}
    \mmd{p, q} = \\ \sum_{k = 0}^n \frac{(1 - \rho)^{n - k} (1 + \rho)^{k}}{2^n} \sum_{\substack{S \subseteq [n] \\ \vert S \vert = k}} \left( \hat P(S) - \hat Q(S)\right)^2, 
\end{multline}
In this representation, we define $\rho = \exp{-1/2\varsigma^2}$, and $\hat P(S)$ as the $S$-Fourier character of the probability distribution $p$. 
The details on how to transform the Hamming-distance into a Fourier-based representation are available in \Cref{app.mmd-fourier}. 
This representation immediately establishes a hierarchy in the Fourier spectrum, since the relative weights differ by a factor $(1 - \rho)^k(1 + \rho)^{-k} \leq \exp{-\rho k}$. 
Fourier components with low order become much more relevant, which also has its reflections on trainability~\cite{rudolph-generative-modelling}. 
As a consequence, the $\MMD$ is particularly useful in the case where the probability distributions of interest have strong components of low-frequency Fourier weights. 

\subsubsection{Product distribution}\label{sec.loss_product}

We now return to the product distributions from \Cref{def.product_distribution} and study the statistical properties of the SD between two instances of this family. Before giving the precise statement, we provide an intuition on why it seems plausible that the SD should concentrate.
The SD is computing squares of differences of several quantities smaller than 1. Thus, in principle, large values SD can be found. However, the concentration properties of the product distribution implies that a very small fraction of $p_C(x)$ will be considerably larger than $2^{-n}$. This fraction decreases exponentially with the number of qubits, translating into very small values of SD.

We formalize the result in the following proposition, which is detailed in \Cref{app:product-distribution}.
\begin{restatable}[SD of product distribution]{proposition}{msdproduct}\label{prop.msdproduct}
Let $p_{\vec a}, q_{\vec b}$ be two product distributions. Then,
\begin{equation}
    \Delta(p_{\vec a}, q_{\vec b}) \leq \frac{\exp{-\ln(3/2) n}}{\delta}
    \end{equation}
with probability at least $1 - \delta$.
\end{restatable}
For the proof, it is only necessary to compute the second moment of $p_C(x)$, which decays exponentially. This is enough to invoke Chebychev's inequality and state this proposition. 

\subsubsection{Pseudo-independent distribution}\label{sec.loss_pseudo}

The pseudo-independent distributions, as we mentioned before, aim to capture all quantum models with classical sampling hardness.
As we have seen in \Cref{sec.def-pseudo-independent}, these distributions anticoncentrate, that is, $\E{}{N^2 p_c(x)^2} \geq \beta'$. This would suffice to set trivial yet easy bounds on the values of the SD. 
In particular, for a given $x$ 
\begin{multline}
    \E{}{(p(x) - q(x))^2} = \Var{p(x) - q(x)} = \\ 
    \E{}{p(x)^2} - \frac{1}{N^2}+ \E{}{q(x)^2} - \frac{1}{N^2} \\ \geq \frac{2 (\beta' - 1)}{N^2},
\end{multline}
implying
\begin{equation}\label{eq:expectation-scaling}
    \E{}{\Delta(p, q)} \in \Omega\left( \frac{1}{N}\right).
\end{equation}

\Cref{eq:expectation-scaling} already implies unfavorable scaling of the loss of the $MMD^2$ on pseudo-independent distributions.
We can extend this into a full concentration inequality as follows.
\begin{restatable}[SD of pseudo-independent distributions]{proposition}{msdpseudo}\label{prop.msdpseudo}
Let $p, q$ be two instances of the pseudo-independent distributions. Then, 
\begin{equation}
    \forall x: \quad \Delta(p_{\vec a}, q_{\vec b}) \leq \frac{\exp{-\mathcal O(n)}}{\delta}, 
    \end{equation}
    with probability at least $1 - \Theta(\delta)$.
\end{restatable}
For the proof, we utilize the upper bound on the deviation of the SD compared to the SD of fully independent distributions, as given by the pseudo-independence condition. Further, the equality of variances of sums of independent variables and sums of variances yields an exact form for the average of the SD. By invoking Markov's inequality we reach the probabilistic statement. We provide the full proof in \Cref{app:loss-function-concentration}.
These results imply concentration also for expectation values. In particular, for any diagonal Pauli string $\mathcal Z$,
\begin{equation}
    \Var{\bra 0 C \mathcal Z C \ket 0} \leq \mathcal O(\exp{-n}).
\end{equation}
Further details are provided in \Cref{def.var_Z} in \Cref{app:independent-distribution}.

To connect the SD to the MMD previously used in several work, we will recover the Fourier description from \Cref{eq.fourier_mmd}. The Parseval-Plancherel identity~\cite{plancherel1910contribution} states that 
\begin{equation}
    \sum_{k = 0}^n \frac{1}{2^n} \sum_{\substack{S \subseteq [n] \\ \vert S \vert = k}} \left( \hat P(S) - \hat Q(S)\right)^2 = \Delta(p, q). 
\end{equation}
In particular, for $\varsigma \rightarrow 0$, the MMD transforms into the SD. 
This admits the interpretation that the MMD we consider is a weighted average over different contributions, but on average equal to the SD. 
Following the same spirit, one could also create a modified version of the SD with added weights. 
In any case, the implicit symmetry of the pseudo-independent distributions implies that, on average, all of these quantities will be similar. 
Only strong biases on the target distribution or the generative model allow for a significant deviation from this mean trend. 
We can thus formulate an analogous result to \Cref{prop.msdpseudo}.  
\begin{restatable}[SD and Fourier-SD concentrate equally]{proposition}{msdmmd}\label{prop.msdmmd}
    \begin{equation}
        \mmd{p, q} \leq \frac{\exp{-\mathcal O(n)}}{\delta} 
    \end{equation}
    with probability at least $1 - \Theta(\delta)$.
\end{restatable}
A proof is available in \Cref{app:independent-distribution}.

\subsubsection{Peaked distribution}\label{sec.loss_peaked}
The peaked pseudo-independent distributions offer clear advantages in terms of concentration of the loss functions. 
The reason is simple. 
The SD measures distances between each $x$ by adding all $p(x) - q(x)$. 
If the distributions are $K$-sparse, with $K \ll N$, two typical distributions will have supports with no or very little overlap. 
Hence, the number of bitstrings of interest will grow as $\mathcal O(K)$, which leads to a much softer upper bound on the averages of the SD.
\begin{restatable}[SD of peaked pseudo-independent distributions]{proposition}{msdpeaked}\label{prop.msdpeaked}
Let $p_{\vec a}, q_{\vec b}$ be two instances of the peaked pseudo-independent distributions. Then, 
\begin{equation}
    \forall x: \quad \Delta(p_{\vec a}, q_{\vec b}) \leq \frac{1}{K \delta}, 
    \end{equation}
    with probability at least $1 - \Theta(\delta)$.
\end{restatable}
The proof can be found in \Cref{app:loss-function-concentration}. The steps for the proof are similar to the case of the pseudo-independent distribution, with the correction of the overlap between supports. This follows a hypergeometric distributions, which can be approximated by a Poisson distribution to obtain the simple bounds. The corrections are negligible in comparison to the trend here shown.

\section{Numerical exploration}\label{sec:numerical-exploration}
While analytical tools provide strong evidences for overall scaling behavior of loss functions, they lack adaptability to real-case scenarios relevant for practitioners. 
For example, one could create circuits that follow the IQP structure but not the assumed distribution of parameters used in the analytical proofs~\cite{bremner2016average, boixo}. 
In this event, theorems do not longer apply due to broken assumptions.
However, similar qualitative behavior to the one shown in the analytical development might be observed. 
In this section, we delve into numerical extensions of the considered models to extend the scope of the analytically tractable examples previously discussed. 

\subsection{Experimental Setup}

In our numerical results, we compute the averages of MSD and \MMD\ of different distributions for number of qubits $n\in[2,\dots,13]$. 
In each experiment, $10^5$ pairs of randomly generated distributions are generated. 

We consider the following distributions. 

\paragraph{Product IQP --} As mentioned in \Cref{sec.def_product}, the product distribution corresponds to a distribution of tensor product states. We will use IQP circuits with only single-qubit operations as a proxy for these states. Those circuits give rise to probabilities of the form
\begin{equation}
    p_{\vec\theta}(x) = \prod_{i=1}^n \vert \bra{x_i} R_x(\theta_i)\ket 0\vert^2,  
\end{equation}
yielding $p_{\vec \theta}(x)$ as a product of sines and cosines, which admit a transformation to match \Cref{def.product_distribution}. 

\paragraph{Matrix Product States (MPS) --} 
MPS are a family of classically tractable states with low entanglement, which have demonstrated success in finding ground states of interesting problems in many-body physics, especially in one-dimensional problems~\cite{cirac2021matrix}. 
The complexity of an MPS is governed by its bond dimension $\chi$, which controls qubit correlations. 
In particular, the entanglement is said to follow an area-law, if a state can be faithfully represented with $\chi \in \mathcal{O}(\poly n)$.  

While MPS were originally designed to compute expectation values, they also allow for a sampling procedure~\cite{ferris2012perfect}. 
Thus, we employ them as a practical extension of our product distribution, incorporating slightly entangled states.
Per definition, setting $\chi=1$ corresponds to the product distribution from \Cref{def.product_distribution}, thus, it is the same as the product IQP outlined above.
To extend the analytic results, we choose $\chi = n$, and generate random MPS using \texttt{quimb}~\cite{gray2018quimb}. 

\paragraph{IQP Circuits --} IQP circuits are one of the paradigmatic examples of classical hardness. 
In the context of this paper, they correspond to pseudo-independent distributions, where each variable follows a Gamma distribution (see \Cref{app:pconnection-haar-dirichlet}). 
IQP circuits are composed by 1) a layer of Hadamard gates, 2) an arbitrary diagonal gate, typically composed by several diagonal operations, and 3) another layer of Hadamard gates.  

For our numerical calculations, we generate random IQP circuits using \texttt{Pennylane}~\cite{bergholm2022pennylaneautomaticdifferentiationhybrid} and the \texttt{IQPopt}~\cite{recioarmengol2025iqpoptfastoptimizationinstantaneous} package. 
We set the \texttt{max\_weight} parameter to 2, corresponding to the maximum number of qubits on which each individual diagonal operation acts. 
The circuits are composed by all possible such pairings without geometrical restrictions.  
The employed operations are rotations around the x-axis with angles $\theta$ sampled uniformly over $[0, 2\pi]$. 

We highlight that this family of IQP does not correspond to a fully pseudo-independent distribution, as this would require an exponential number of independent angles. 
Despite that, this restricted version already suffices to prove classical hardness on sampling, which is why we employ the model as a practically useful and considerably cheaper approximation of the full family.
In fact, restricted versions suffice to prove classical hardness on sampling \cite{bremner2016average}. 

\begin{figure}

    \includegraphics[width=.9\linewidth]{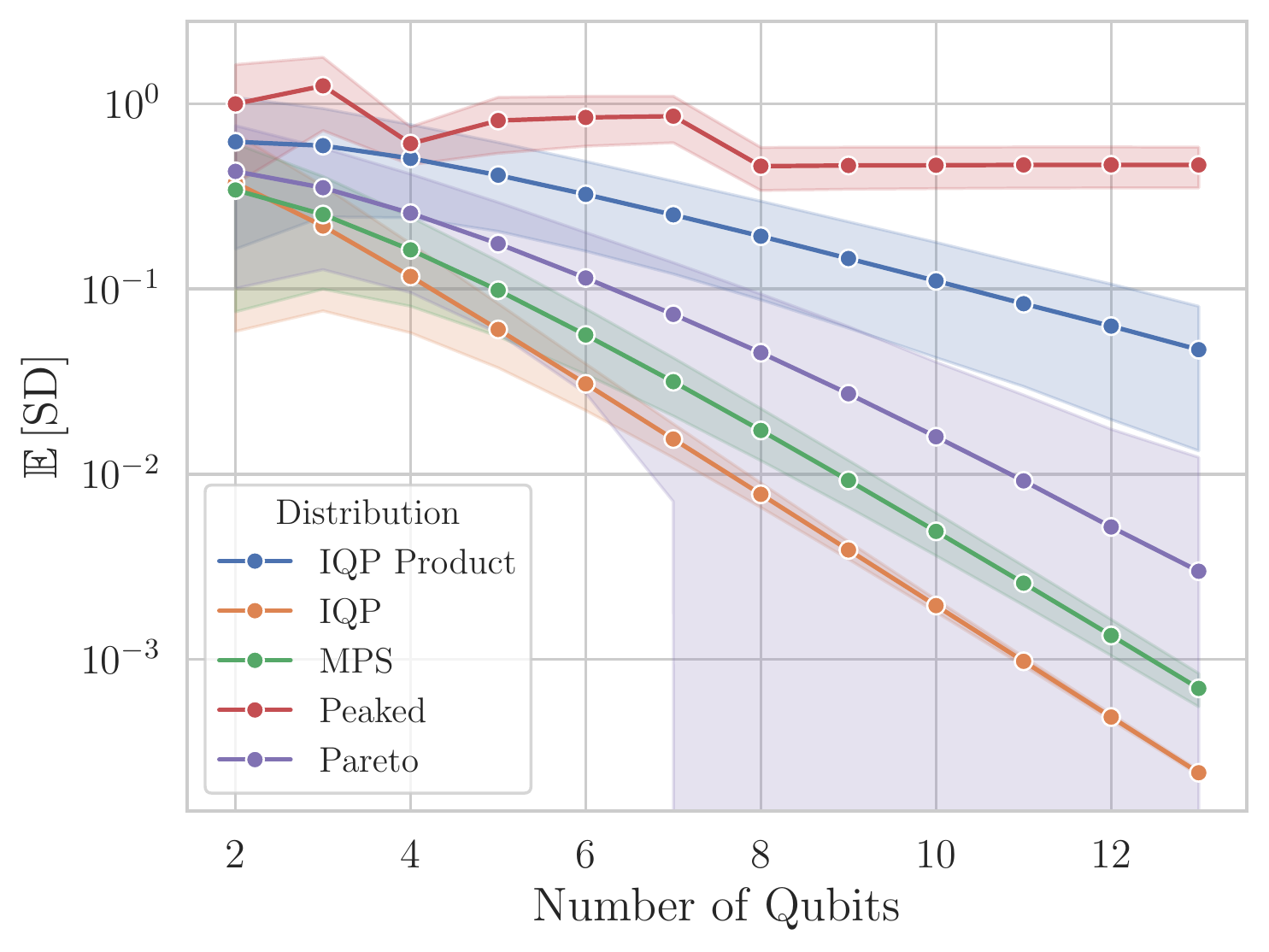}
    \caption{Average of the SD for the considered models. As expected, all of the averages decay exponentially with the number of qubits, with the exception of the peaked distributions. The analytical bounds are satisfied in all the considered examples. For the Pareto distribution, exponential decay is also observable, even though the analytical bound is meaningless. }
    \label{fig:numerics-msd}

    \includegraphics[width=.9\linewidth]{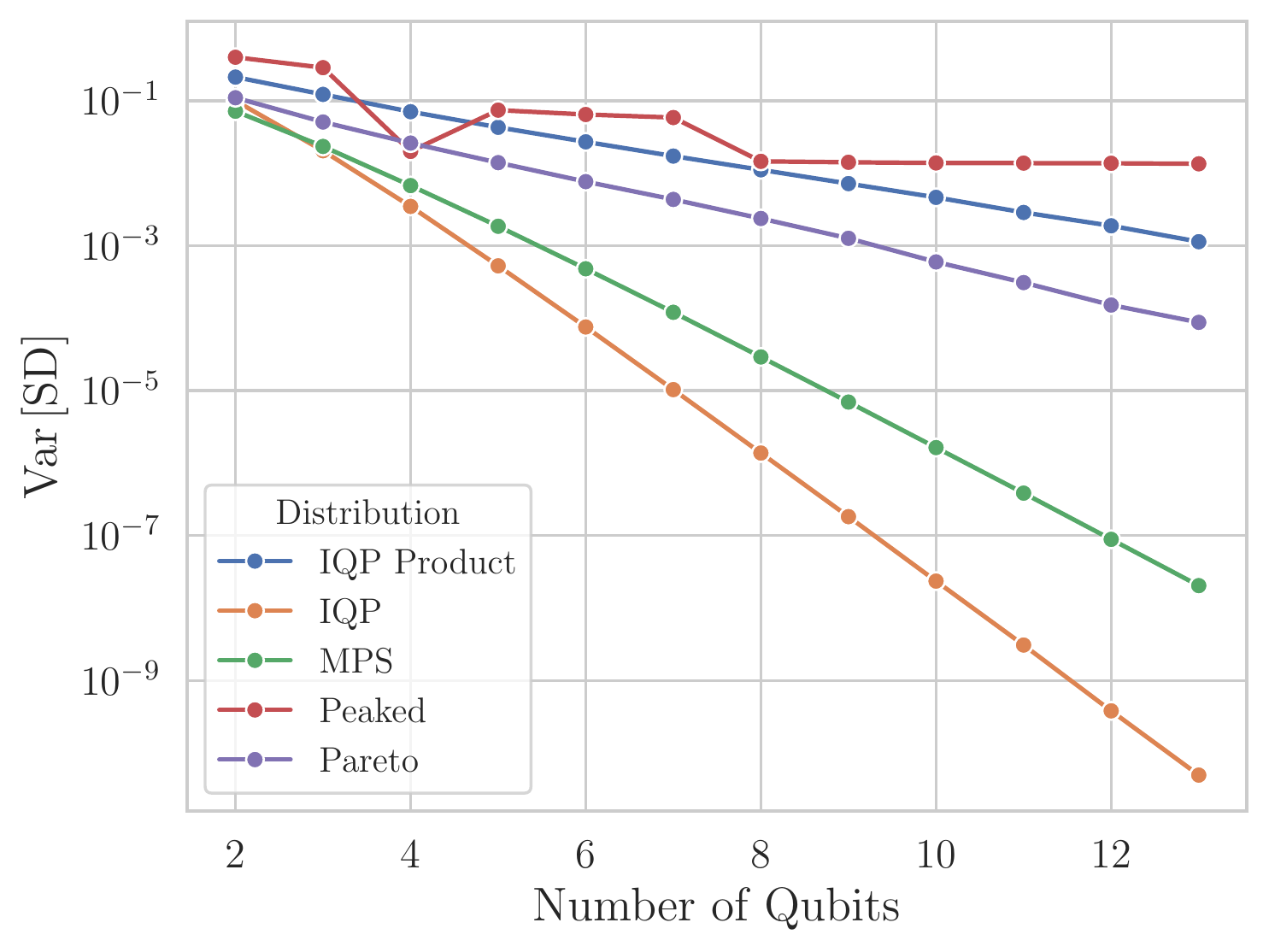}
    \caption{Variance of the SD for the considered models. The variances also decay exponentially with the number of qubits. In the case of IQP circuits, its strongly decaying bias is a consequence of pseudo-independence. }
    \label{fig:numerics-msd-variance}
\end{figure}

\begin{figure}
        \includegraphics[width=.9\linewidth]{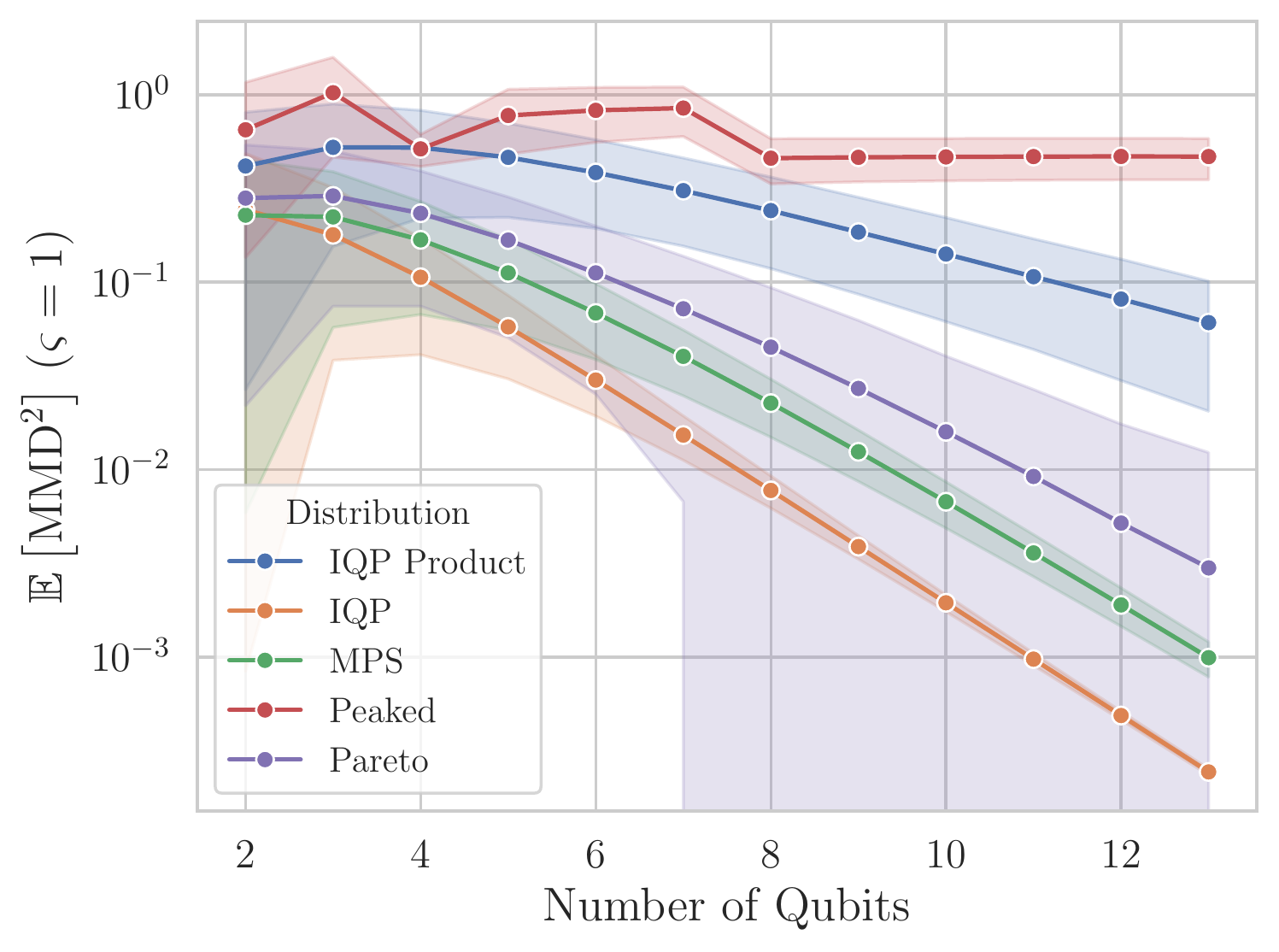}
        \caption{Average of the $\MMD$ with $\varsigma\in O(1)$ for all considered models. The behavior is similar to that of SD, with exponential decay in all cases, highlighting the randomness of the probability distributions of interest. }
        \label{fig:mumerics-mmd}

        \includegraphics[width=.9\linewidth]{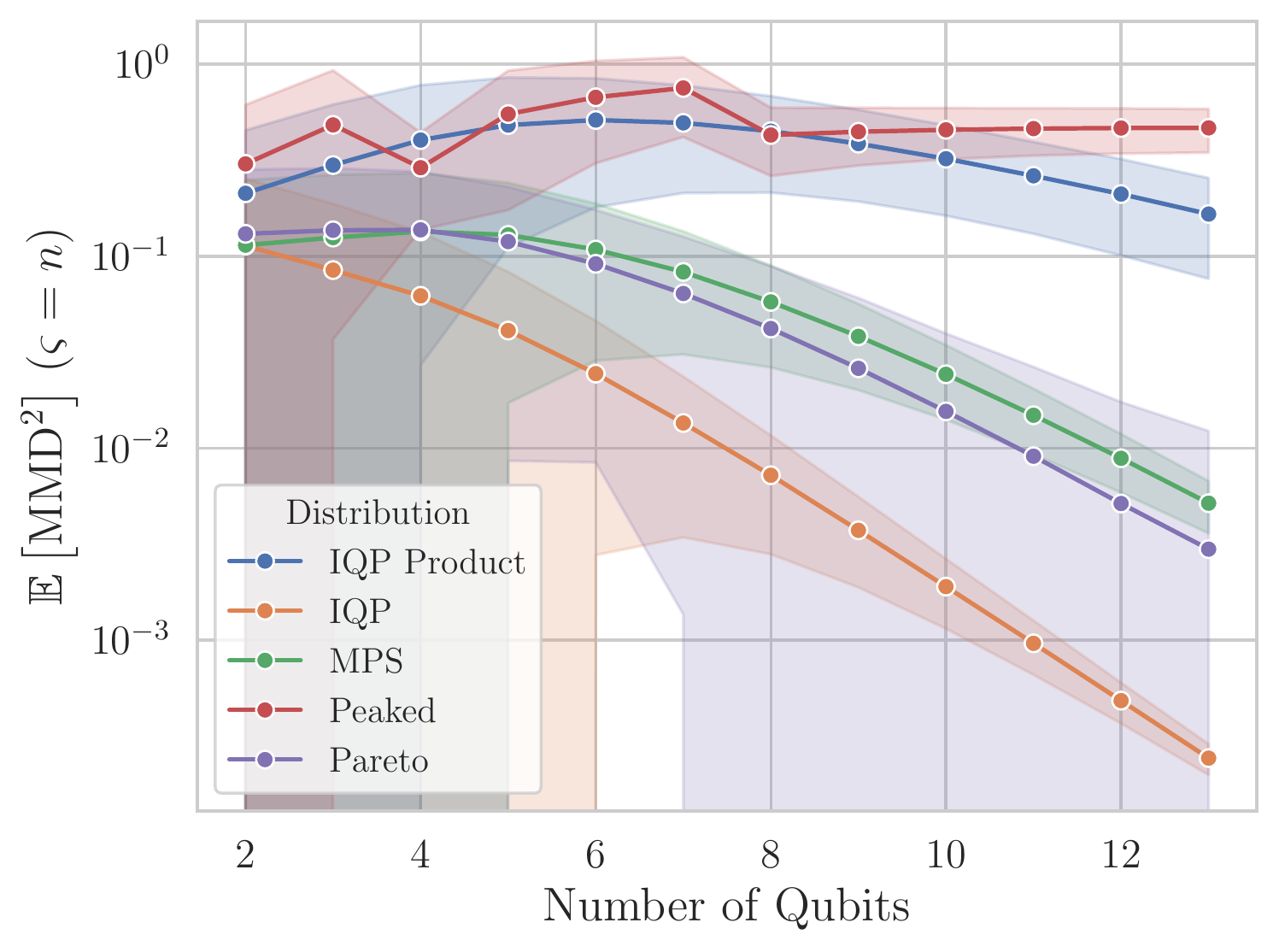}
        \caption{Average of $\MMD$ with $\varsigma\in O(n)$ for all considered models. In this experiments, IQP product does not decay as strongly as other approaches. This behavior is due to the special properties of $\MMD$, which is highly sensitive to low-order correlations. Those are dominant for product circuits. }
        \label{fig:mumerics-mmdoon}
\end{figure}

\paragraph{Pareto Distribution --} 
To extend our analysis of pseudo-independent distributions, we additionally consider a Pareto distribution, that is, a random distribution whose density function decays as $\mathcal{O}(Y^{-\alpha})$ (see \Cref{def:pareto}). 
This distribution is interesting as an example of a pseudo-independent distribution that does not decay as quickly as that of Haar-random states. The variance is given by
\begin{equation}
    \sigma = \frac{\alpha}{(\alpha - 1)^2 (\alpha - 2)}.
\end{equation}
In our numerical example we choose the value $\alpha = 2$, which eliminates any tightness in the analytical bounds, while maintaining a well-behaved probability distribution for unbounded values of $Y$. 
The interest in analyzing this distribution is to gain insight into the behavior of the loss functions, where analytical guarantees fail. 
A decay scaling as $\mathcal{O}(N^{-1})$ would discourage the use of pseudo-independent distributions as generative models generally.

To the best of our knowledge, there is no known corresponding quantum circuit. 
However, this is not strictly required per our formalism as it suffices to have a description of the probability distributions. 
We use \texttt{scipy}~\cite{2020SciPy-NMeth} to generate the random variables accordingly.

\paragraph{Peaked Circuits} Finally, we generate peaked probability distributions by considering IQP circuits on a subset of $\log(n)$-qubits, and randomly distribute them in the $n$-qubit Hilbert space. This setting corresponds to the pseudo-independent peaked distributions from \Cref{def.peaked_distribution}, with $K \in \Theta(n)$.

\subsection{Results}
We first analyze the behavior of the SD for the considered distributions in \Cref{fig:numerics-msd}. As expected, the IQP and IQP-product states show an exponential decay in the expectation value of the SD, in agreement to our theoretical predictions. Further, the MPS distribution serves as an intermediate step between IQP and IQP product. For the case of peaked circuits, the scaling follows the expected trend, with artificial steps due to the rounding of $\log(n)$. The Pareto distribution we consider does not yield any behavior significantly different from the other considered ones. In spite of the inexistence of a bound in $\E{}{\Delta(p, q)}$, the numerical experiments show a similar exponential behavior as IQP or IQP product. We highlight an interesting fact of these results. The IQP circuits do not only show an exponential decay on the average, but also strong concentration properties as the number of qubit increases, in agreement with our assumption of pseudo-independence. This observation is more clearly visible in \Cref{fig:numerics-msd-variance} 

In alignment with previous works \cite{rudolph-generative-modelling}, we explore the $\MMD$ for $\varsigma\in O(1)$ in \Cref{fig:mumerics-mmd}, and for $\varsigma\in O(n)$ in \Cref{fig:mumerics-mmdoon}. For $\varsigma = 1$, we observe similar trends as the ones explained for the MSD, showing that their average behaviors are identical, and that $\MMD$ will not provide on average trainability advantage over the SD. For $\varsigma = n$, there is a change for IQP product and MPS states. As discussed in \Cref{sec.preliminaries_loss_functions}, our $\MMD$ is sensitive to low-order correlations, present in states with low entanglement. It hints that these models might be easier to train than IQP, but also classical methods exists for sampling. 

\section{Verification}\label{sec.verification}

In this manuscript, the discussion hovers over the concept of loss-function concentration as a proxy to trainability. 
We utilize this narrative for its connection to barren plateaus \cite{cerezo2025does, larocca2024review}, where exhaustive research established fundamental trade-offs between classical advantage and trainability. 
In the context of GML, trainability means finding a candidate distribution $p$ that best approximates $q$. 
A bare minimum requirement to accomplish this task is being able to distinguish $p$ from $q$. 
This task is known as \textit{verification} \cite{valiant-valiant, hangleiter}.
In this section, we provide an interpretation of our findings under the lens of verification. 

In quantum computing, verification is phrased as the task to certify that a quantum computation with promised advantage is actually a quantum computation and not a classical approximation of it. 
In the common setup of verification, a full description of a target distribution $q$ is given, together with sampling access to a trial distribution $p$.
We need to distinguish the cases
\begin{equation}
    \left\{ \begin{matrix}p = q & \\ 
    \|p-q\|_1& \ge\epsilon \end{matrix} \right. \quad , 
\end{equation}
with probability at least $2/3$\footnote{As in many examples, this success probability only needs to be bounded away from $1/2$, and it can be boosted to arbitrarily close to $1$ using repetition}, where $\norm{\cdot}_1$ is the Total Variation Distance (TVD).
The minimal sample complexity to distinguish two distributions is tightly related to the norm of its truncated probability vector~\cite{valiant-valiant}. 
For quantum distributions with classical sampling hardness, anticoncentration is a key ingredient for advantages, but also an insurmountable obstacle for verification~\cite{hangleiter}.
In short, exponentially many samples are needed due to the output distributions being \textit{almost but not exactly} uniformly distributed. 

To bridge the gap between training and verification, we need to circumvent the need for a complete description of $q$, which is typically available in the form of samples in any ML setting. 
However, for any distance, the triangle inequality implies 
\begin{equation}
\left\vert \norm{p - u} - \norm{q - u}\right\vert \leq \norm{p - q} \leq \norm{p - u} + \norm{q - u}, 
\end{equation}
where $u$ is the uniform distribution, from which we have a complete description. 
Therefore, independent verification of the distributions $p$ and $q$ suffice to provide bounds on the distances between them which is needed for training, with constant overhead. Notice that $\Vert p - u \Vert$ (resp. $q$) is related to the computation of $\E{}{p(x)^2}$, since $u$ is the average distribution of all $p$, thus connected to anticoncentration. 

To complete the connection, we can consider the TVD as a loss function. In particular, it holds that
\begin{equation}
    \Vert p - q\Vert_1 \leq \sqrt{N \Delta(p, q)}, 
\end{equation}
which could open the possibility to use the TVD as a loss function for generative models due to its exponentially larger values. 
The bound implies that $\E{}{\norm{p - q}_1} \in \mathcal O(1)$, and prevents us from using the mean to show concentration. 
However, this naive approach will not change the overall picture. 
The value of the loss function itself is not relevant for training, but its variance. 
From verification, this variance must decay exponentially with the number of qubits \cite{hangleiter}. 
Also, the concentration properties of the SD, stemming from small average values, impose variances of the same order in the TVD. In particular, for pseudo-independent distributions
\begin{equation}
    \Var{\Vert p - q \Vert} \in \mathcal O(\exp{-n}), 
\end{equation}
as detailed in \Cref{le.var_abs_diff}, \Cref{app:loss-function-concentration}. We also provide a numerical exploration of the TVD for our experiments in \Cref{app:extension-numerics}.

\subsection{Peaked circuits}\label{sec.peaked_circuits}
Since verification is a necessary previous step to train a generative model, it is interesting to look into circuits whose outputs can be verified, such as peaked circuits \cite{zhang2025complexityhardnessrandompeaked, gharibyan2025heuristic}. 
To the best of our knowledge, such circuits are, however, hard to find. 
An example is to concatenate two random circuits and post-select the second one so that the final state is peaked \cite{aaronson2024verifiablequantumadvantagepeaked}, a procedure that brings exponential overhead.

Peaked circuits do not admit estimation of $p(x)$ to multiplicative precision in general.
Otherwise, Shor's algorithm, an example of a peaked circuit \cite{Shor_1997}, could be dequantized to factorize numbers with classical machines. 
Regarding sampling, there must be a classical algorithm that produces samples coming from the distribution created by any peaked circuit.
This is the case because peaked circuits concentrate their outputs in a small support, and a classical algorithm could in principle store in memory all of the outputs to sample them efficiently. 
However, because estimating $p_C(x)$ is hard, no purely classical recipe to find this classical sampling surrogate seems possible. 

\subsection{Distinguishing distributions with \MMD}\label{sec.estimation}
Finally, it is possible to show that similar constraints to the ones found in verification apply to estimating a loss function and distinguishing between distributions.
The $\MMD$ has been considered for this task in previous works~\cite{JMLR:v13:gretton12a}.
While the true $\MMD$ is not accessible, it is possible to compute an unbiased
estimator of it, \estMMD, from $l$ samples from the candidate distribution $p$, and $m$ data points from $q$.
The estimator converges as follows,
\begin{multline}\label{eq.statistical_test}
        \prob{|\estMMD - \MMD| > t } \leq \\ \exp{-\frac{t^2 (m + l)}{4K^2}}, 
\end{multline}
where $K$ is an upper bound to the absolute value of the kernel, for the Gaussian kernel $K \leq 1$. 
This, in turn, allows us to devise a statistical test for equality.
\begin{corollary}[Statistical test for \estMMD]\label{cor.stat-test-mmd}
    A hypothesis test of level $\alpha$ for the null hypothesis $p=q$ has the acceptance region
    \begin{equation}\label{eq:statistical-test}
        \estMMD \le \sqrt{\frac{8 \log(\alpha^{-1})}{m + l}}.
    \end{equation}
\end{corollary}
Obtaining sufficiently large values for the estimator is then necessary to distinguish between $p$ and $q$. For exponentially vanishing values of the true $\MMD$, such distinction will be very unlikely with $m,l\in\poly{n}$. The vast majority of statistical tests will return $p=q$, making the search for an optimal $p$ unfeasible. 
This is the typical scenario for a distribution probability showing anticoncentration.
We refer to \Cref{app:bounds-mmd} for further details on these results.

\section{Implications for QML}\label{sec.implications_qml}
The results detailed in this manuscript have obvious implications for GQML.
Certain models with quantum advantage, such as IQP or RQC, are not trainable, and suffer from the same problems as many other variational models \cite{McClean_2018, cerezo2025does}. 
Only in the assumption that a strong bias exists to achieve an efficient and informative training the model becomes useful. 
Otherwise, training will not be possible, as most of the possible output distributions are \textit{almost but not exactly} uniform.

To escape from this phenomenon, one could first consider a target probability distribution $q$ which is an outlier of the generative model, hence
\begin{equation}
    \Delta(q, u) \gg N^{-1}, 
\end{equation}
implying that $q$ is easily distinguishable from $u$ (the uniform distribution). In this case, the triangle inequality implies
\begin{multline}
    \left( \Delta(q, u)^{1/2} - \Delta(p, u)^{1/2}\right)^2 \\ \leq \Delta(p, q) \leq \\ \left( \Delta(q, u)^{1/2} + \Delta(p, u)^{1/2}\right)^2.
\end{multline}
Therefore, $\Delta(p, q)$ shows concentration inherited from the concentration on $\Delta(p, u)$, even though the average is increased to $\approx \Delta(q, u)$. 

Biased models have been suggested in previous works \cite{recioarmengol2025trainclassicaldeployquantum, rudolph-generative-modelling}, where data coming from probability distributions with sparse Fourier spectrum could be trained using $\MMD$. 
A signature for this phenomenon was seen in \Cref{sec:numerical-exploration}. 
This is the case for a product distribution, whose Fourier transform is given by
\begin{equation}
    \hat P_{\vec a}(S) = \prod_{i = 1}^n (2a_i - 1)^{S_i}, 
\end{equation}
which implies that low-weight $S$ have geometrically larger Fourier components than high-weight $S$. 
This nicely aligns with the Gaussian kernel, and the strong bias will make these models trainable. 
In fact, 
\begin{align}
    \varsigma \in \mathcal O(n) \rightarrow & \E{}{\Delta(p, q)} \in \Omega(1), \\
    \varsigma = n \rightarrow & \E{}{\Delta(p, q)} \geq 2 e^{-1/6}
\end{align}
It is, however, important to highlight that these instances do not belong to the average case of models with quantum advantage rooted in anticoncentration, as they are classically representable, as discussed in \Cref{sec.def_product}. 

At the core, ML consists in learning a distribution from data, that is, approximating it up to a certain accuracy. 
We remark that distribution learning is not equivalent to verification, in the sense of distinguishing two instances. 
In particular, there exist distribution classes that are learnable but not verifiable \cite{goldreich}, for instance the product distribution. 
Hence, the impossibility to verify a probability distribution does not immediately preclude learning it. 

Finally, ML has lately investigated the \textit{simplicity bias} \cite{shah2020pitfallssimplicitybiasneural} as a feature that allows for generalization of the training dataset. 
This bias is necessary to ensure that a small training dataset must be representative enough to allow for the learning of the entire distribution. 
Probability distributions exhibiting anticoncentration do not comply with the simplicity bias. 
In light of our results, it is impossible to use $\poly n$ many samples to learn its underlying distribution, assuming a typical output from a generative model with anticoncentration. 
Consider $p_C(x) \approx 2^{-n}$ of this kind. 
A dataset of size $\poly{n}$ sampled from this distribution will be a list of unique samples. 
The available information does not suffice to distinguish the underlying distribution between the uniform or the true $q$, or even from a somewhat uniform distribution with support over a subset of bitstrings. 
In consequence, samples will only be useful if the underlying probability distribution is sufficiently peaked, that is, concentrated, at the expense of losing the sampling hardness that motivates the use of such a quantum model.

\section{Conclusion}\label{sec.conclusion}
In this paper, we have analyzed the statistical properties of typical loss functions of quantum generative models, in particular those exhibiting anticoncentration and quantum advantage. 
To this end, we gave three models of families of probability distribution, each one corresponding to a typical example of quantum circuits, namely i) product distributions, related to tensor-product circuits, ii) pseudo-independent distributions, related to IQP or RQC, and iii) peaked pseudo-independent distributions, related to peaked circuits. 
We showed that typical loss functions of generative models corresponding to the first and second kind of distributions, have exponentially vanishing average values. 
Therefore, it is extremely unlikely that any statistical test can even distinguish between two such distributions and the models become effectively untrainable.
Our analytical results are complemented with numerical experiments to bridge the gap for a broader set of quantum circuits.

The immediate consequence is that quantumly advantageous generative models whose advantage stems from anticoncentration will not be usable in practice. 
The reason for our observation lies at the foundation of this kind of quantum advantage, where multiplicative errors on approximating certain exponentially-vanishing quantities are required. 
As a result, the typical outcome of such a model will be \textit{close but not exactly} uniformly distributed, which questions the applicability of such models. 

Our results prevent the use of such circuits for GQML, unfortunately leading to the loss of a candidate for advantageous generative models.
However, there are other possible sources for advantages that survive our analysis. 
A clear candidate are peaked circuits, whose distributions strongly concentrate on small supports whose probability densities $p_C(x)$ are know to not be classical estimable. 
Nevertheless, the properties of peaked circuits might make their output distributions classically learnable from data. 
Another possible candidate for useful and, potentially, trainable models are models targeting datasets with a strong bias. However, several examples throughout the manuscript show that keeping the sampling advantage while biasing the model is a non-trivial task. 

The results here presented encourage us to reflect on the applicability of QGML. 
Classical hardness as stated in anticoncentrated models necessarily imply trainability issues. 
However, interesting generated data should obey a specific purpose, and not be approximately uniform.
In line with this discussion, generative models should be addressed from an angle complementary to classical hardness from multiplicative errors. 

\section*{Code and data availability}

The code used in this work is available in the GitHub repository at~\cite{githubcode}.

\acknowledgments
The authors would like to thank Yuxuan Zhang, Maria Schuld and Vincenzo De Maio for valuable discussions. SH is a recipient of the Andreas Dieberger-Peter Skalicky scholarship and acknowledges financial support by the netidee funding campaign (grant 7728). This work has been partially funded through the Themis project (Trustworthy and Sustainable Code Offloading), Austrian Science Fund (FWF): Grant-DOI 10.55776/PAT1668223, the Standalone Project Transprecise Edge Computing (Triton), Austrian Science Fund (FWF): P 36870-N, and by the Flagship Project HPQC (High Performance Integrated Quantum Computing) \# 897481 Austrian Research Promotion Agency (FFG). APS acknowledges the Swiss National Science Foundation for its support through the grant TMPFP2\_234085.

\bibliography{refs.bib}

\onecolumngrid
\newpage
\appendix

\section{Preliminaries}
We begin with several definitions, which are used in subsequent proofs.

\begin{definition}[Gamma function~\cite{abramowitz1965handbook}]\label{def:gamma-func}
    The Gamma function is defined as 
    \begin{equation}
        \Gamma(z) = \int_0^{\infty} t^{z-1}e^{-t}dt
    \end{equation}
    for any real $z > 0$. For positive integers, it takes the form

    \begin{equation}
        \Gamma(z) = (z-1)!
    \end{equation}
\end{definition}

\begin{definition}[Incomplete lower Gamma function~\cite{abramowitz1965handbook}]\label{def:incomplete-lower-gamma}
    The incomplete lower Gamma function is defined as

    \begin{equation}
        \gamma(a,x) =\int_0^{x} e^{-t}t^{a-1}dt
    \end{equation}

    with $a\in\mathbb{C}, \Re(a)>0$. 
\end{definition}

\begin{definition}[Gamma distribution]\label{def:gamma-distribution}
    The probability density function of a Gamma distribution with shape $\alpha > 0$ and scale $\theta > 0$ is defined as 

    \begin{equation}
        f_{\alpha,\theta}(x) = \frac{1}{\Gamma(\alpha)\theta^\alpha}x^{\alpha - 1}e^{-x/\theta}.
    \end{equation}
    Its cumulative distribution is proportional to the incomplete lower Gamma function.
\end{definition}

\begin{definition}[Dirichlet distribution~\cite{multivariate}]\label{def:dirichlet}
	An $N$-dimensional random variable $\vec x$ follows a Dirichlet distribution if it follows a probability density function given by
	\begin{equation}
		f(\vec x) = \frac{\Gamma(\sum_i\alpha_i)}{\prod_{i = 1}^N \Gamma(\alpha_i)} \prod_{i = 1}^N x_i^{\alpha_i - 1}, 
	\end{equation}
	for $\alpha_i \geq 0$.
	The marginal probabilities follow Beta distributions with $\alpha_0 = \sum_{i = N} \alpha_i$.
\end{definition}

\begin{definition}[Beta distribution~\cite{bailey1992distributional}]\label{def:beta}
A random variable $x$ follows a Beta distribution if it follows a probability density function given by
\begin{equation}
    f(x) = \frac{\Gamma(\alpha + \beta)}{\Gamma(\alpha) \Gamma(\beta)} x^{\alpha - 1} (1 - x)^{\beta - 1}, 
\end{equation}
    where $\Gamma(\cdot)$ is the Gamma function. Beta distributions are the marginals of Dirichlet distributions. 
\end{definition}

\begin{definition}[Symmetric Dirichlet distribution]\label{def:sym-dirichlet}
	A Dirichlet distribution is symmetric if $\vec\alpha = \alpha \vec 1$. 
\end{definition}

\begin{definition}[Poisson distribution]\label{def:poisson}
    A random variable $X$ is said to follow a Poisson distribution with $\lambda > 0$ if its probability mass function is given by

    \begin{equation}
        \Pr(X=k) = \frac{\lambda^k e^{-\lambda}}{k!}
    \end{equation}
\end{definition}

\begin{definition}[Pareto distribution]\label{def:pareto}
    A random variable $X$ is said to follow a Pareto distribution with $\alpha > 1$ if its probability mass function is given by
    \begin{equation}
        \Pr(X=p) = \frac{\alpha}{(1 + x)^{\alpha + 1}}.
    \end{equation}
\end{definition}

\begin{definition}[Gini coefficient]\label{def:gini}
    The Gini coefficient of a positive random variable $Y$ is given by
    \begin{equation}
        G(Y) = \frac{\E{Y,Y'}{\vert Y - Y'\vert }}{2 \E{}{Y}},
    \end{equation}
    with $Y'$ being a random variable independent from $Y$ but following the same probability distribution. 
\end{definition}

\section{Connection between pseudo-independent distributions and Haar-random states}\label{app:pconnection-haar-pseudo}

\subsection{Connection between Dirichlet distributions and Haar-random states}\label{app:pconnection-haar-dirichlet}
It is a known fact that the probabilities $p(x)$ of measuring a bitstring $x$ after preparing a Haar-random state follow a symmetric Dirichlet distribution \cite{boixo, olkin1964multivariate}, with $\alpha = 1$. This implies that the marginal probability sampling each different $x$ follows a probability distribution given by
\begin{equation}
    \prob{p(x) = y} = N (1 - y)^{N - 2}, 
\end{equation}
which is a Beta distribution~\cite{bailey1992distributional}. 
The Porter-Thomas distribution commonly discussed in the supremacy experiments is just an approximation of this distribution
\begin{equation}
    \prob{p_C(x) \geq y} = N \int_{y}^1 (1 - t)^{N - 1} dt = (1 - y)^N \lesssim \exp{-N y},
\end{equation}
Notice that the Porter-Thomas distribution is by definition incorrect since $y$ is unbounded. However, it still returns accurate descriptions. 

It is important to give relevance to the marginals of the distributions with respect to the distribution itself. It is somewhat accepted in the folklore that a Porter-Thomas shape is the signature for \textit{quantumness} \cite{boixo}. 
However, sampling a constant number of elements from any distribution admitting an analytical description is efficient \cite{bouland2018quantum}. 
We remark that this behavior does not suffice to identify the probability distribution as sourced from a quantum device. 
In fact, sampled elements from marginals do not even constitute a probability distribution because they are not constrained to sum up to one. 
We conjecture that the true source of \textit{quantumness} is not visible in the marginals, but in the correlations of the sampling probabilities. 
These correlations must be exponentially small, and thus only visible when exponentially many samples are obtained. 
Along these lines, recent works on lower bounds for the query complexity of probability distributions might open interesting directions to shed light on classical sampling methods \cite{chewi2023query}.  

\subsection{Connection between pseudo-independent and Dirichlet distributions}\label{app:pconnection-dirichlet-indepdent}   

\begin{theorem}
    Drawing $X_i$ i.i.d. according to a Gamma distribution with shape $\alpha$ and rate $1$ from \Cref{def:gamma-distribution} results in the pseudo-independent random variables distributed according to a symmetric Dirichlet distribution from \Cref{def:sym-dirichlet}.
\end{theorem}

\begin{proof}
    We consider i.i.d. random variables drawn from a Gamma distribution with shape $\alpha$ and rate 1 $X_1, \dots, X_N \sim^{i.i.d.} \Gamma(\alpha, 1)$ with their densities given as in \Cref{def:gamma-distribution}.

    We consider the normalized vector 

    \begin{equation}
        Y_i = \frac{X_i}{\sum_i X_i}
    \end{equation}

    and let $S = \sum_i X_i$. The joint probability distribution of the random variables is given as 

    \begin{equation}
        p(X_1,\dots,X_N) = \prod_i p(X_i) = \prod_i \frac{1}{\Gamma(\alpha)} X_i^{\alpha - 1}e^{-X_i}
    \end{equation}

    We now shift the perspective and consider $X_i = Y_iS$. For $i=N$, we instead look at $X_N = S(1-\sum_{i\in [1,N-1]}Y_i)$. We shift from $p(X_1,\dots,X_N)$ to $p(Y_1S,\dots, Y_{N-1}S, (\sum_{i\in[1,N-1]}Y_i)S$, and consider the Jacobian first.

    \begin{equation}
        J = \frac{\partial(X_1,\dots,X_n)}{\partial(Y_1,\dots,Y_{N-1},S} 
    \end{equation}
    Its determinant will encode information on how to transform variables from one to another description. 
    We can now derive the partial derivatives for $i\in[1,N-1]$ as $\frac{\partial X_i}{\partial Y_i} = S$, $\frac{\partial X_j}{\partial Y_i} = 0$, for $i \ne j$ and  $\frac{\partial X_i}{\partial S} = Y_i$. For $X_N$ we see $\frac{\partial X_N}{\partial Y_i} = -S$ and $\frac{\partial X_N}{\partial S} = Y_n$. By standard properties of the determinant, we can pull out factor of $S$ from one row, and multiply it onto the determinant afterward. We do this for all $N$ columns, leaving the Jacobian as 

    \begin{equation}
        J = \begin{bmatrix}
            1 & 0 & \dots &  0 & \frac{Y_1}{S} \\
            0 & 1 & \dots &  0 & \frac{Y_2}{S} \\
            \vdots & \vdots & \ddots & \vdots & \vdots \\
            0 & 0 & \dots & 1 & \frac{Y_n}{S} \\
            -1 & -1 & \dots & -1 & \frac{Y_n}{S}
        \end{bmatrix}
    \end{equation}

    We add all rows to the last row, resulting in an upper triangular matrix, whose determinant is equal to the product of diagonal entries, which is $1/S$. Multiplying by the previous factor of $S^N$ reveals the determinant as 
    \begin{equation}
        \left\vert \det(J) \right\vert = S^{N - 1}
    \end{equation}.

    We now turn to the joint probability distribution as

    \begin{align}
        p(Y_1,\dots,Y_{N},S) & = \prod \frac{1}{\Gamma(\alpha)}(SY_i)^{\alpha-1}e^{-SY_i} S^{N-1}\\
        &= \frac{1}{\prod_{i=1}^{N}\Gamma(\alpha)}S^{N-1}S^{N(\alpha - 1)}e^{S} \prod Y_i^{\alpha-1}
    \end{align}

    We can now split this into two parts, one dependent only on $S$ and one on $Y$. 

    \begin{align}
        &p(S) = \frac{1}{\Gamma(N\alpha)}S^{N\alpha-1}e^S \\
        &p(Y_1,\dots,Y_n) = \frac{\Gamma(N\alpha)}{\prod_{i=1}^{N}\Gamma(\alpha)} \prod Y_i^{\alpha - 1}
    \end{align}

    It follows that the joint distribution of the pseudo-independent random variables follows a symmetric Dirichlet distribution $\{Y_1,\dots,Y_N\} \sim \Dir$, and $S$ follows the Gamma distribution $S \sim \Gamma(N\alpha, 1)$.
\end{proof}

\section{Loss function concentration}\label{app:loss-function-concentration}

\subsection{Product Distribution}\label{app:product-distribution}
We consider the product distributions from \Cref{def.product_distribution}.

\begin{lemma}[Marginals of the product distribution]\label{le.marginals_product}

Let $p_{\vec a}$ be a product distribution, with $\vec a \sim U^{(n)}$. The probability distribution of $p_{\vec a}(x)$ is given by 
\begin{equation}
\prob{p(x) = y} = \frac{\ln(1 / y)^{n - 1}}{(n-1)!}    
\end{equation}
for any $x \in \{ 0, 1\}^n$ and $0 \le y \le 1$. 
\end{lemma}

\begin{proof}
    First, we note that the probability distribution has a strong symmetry by exchanging $0$ and $1$ in the variable $x$ and elements in $\vec a$. That is, the probability distribution for a given value of $a_i$ perfectly matches that of $(1 - a_i)$ under a transformation $x_i \rightarrow x_i \oplus 1$. Since we assume a uniform distribution, both elements are equally probable, and the average over $\vec a$ implicitly performs an average over $x$. Hence, all marginals are the same, and we just need to focus on one of them, for instance $\vec x = 0\ldots 0$. The condition is achieved by 
    \begin{equation}
        \prod_{i = 1}^n a_i = y. 
    \end{equation}
    We perform the change of variables
    \begin{equation}
        a_i = \exp{-z_i}, \qquad z_i \geq 0, 
    \end{equation}
    which transforms the condition into 
    \begin{equation}
        \sum_{i} z_i = \ln(1 / y). 
    \end{equation}
    This condition on $\mathbb R^n$ describes a simplex over $n$ dimensions, with volume 
    \begin{equation}
        \operatorname{Vol}_{\sum_i z_i = \ln(1 / y)} = \frac{\ln(1 / y)^{n - 1}}{(n - 1)!}. 
    \end{equation}
\end{proof}

\begin{lemma}[Marginals of the product distribution - II]\label{le.marginals_product2}
  In the assumptions of \Cref{le.marginals_product},
  \begin{equation}
\prob{p(x) \geq \frac{y}{2^n}} = \frac{\gamma(n, n \ln 2 - \ln(y))}{\Gamma(n)},     
\end{equation}
with $\Gamma(\cdot)$ the Gamma function from \Cref{def:gamma-func}, and $\gamma(\cdot, \cdot)$ the lower incomplete gamma function from \Cref{def:incomplete-lower-gamma}.
\end{lemma}

\begin{proof}
    Following the procedure from the previous lemma, the new condition is given by
        \begin{equation}
        \prod_{i = 1}^n a_i \geq \frac{y}{2^n}, 
    \end{equation}
    hence
    \begin{align}
        \prob{\prod_{i = 1}^n a_i \geq \frac{y}{2^n}} &= \int_{\prod_{i = 1}^n a_i \geq \frac{y}{2^n}} \prod da_i \\ 
        & = \int_{\sum z_i\leq \ln(2^n/y)} \exp{-\sum_i z_i} \prod_{i = 1}^n d z_i \\
        &= \int_{0}^{n \ln(2) - \ln(y)} \exp{-s} \frac{s^{n - 1}}{(n-1)!} ds.  
    \end{align}
    By definition, this is the lower incomplete gamma function from \Cref{def:incomplete-lower-gamma}. 
\end{proof}

\concentrationproduct*

\begin{proof}
We begin from \Cref{le.marginals_product2} with the definition of the incomplete gamma function, and a rewriting for integer numbers in the form \cite{arfken2013mathematical}
\begin{equation}
    \frac{\gamma(n, x)}{\Gamma(n)} = \exp{-x} \sum_{k = n}^\infty \frac{x^k}{k!}. 
\end{equation}
This matches the probability of obtaining more than $n$ events in a Poisson distribution from \Cref{def:poisson}, with $\lambda = x$. The Chernoff bound for the Poisson distribution is given by~\cite[p.50]{Vershynin_2018}
\begin{equation}
    \prob{X \geq a} \leq \left(\frac{\lambda}{a}\right)^{a} \exp{a - \lambda}. 
\end{equation}
In our case, 
\begin{align}
    \lambda & = n \ln(2) - \ln(y) \\
    a & = n, 
\end{align}
hence

\begin{equation}
    \prob{p(x) \geq \frac{y}{2^n}} \leq \left( \frac{n \ln(2) - \ln(y)}{n} \right)^{n} \exp{n - \ln(2) n + \ln(y)}
\end{equation}

Simplifying the first term on the right hand side, we obtain
\begin{align}
    \left(\frac{n\ln(2) - \ln(y)}{n}\right)^n & = \left(\ln(2) - \frac{\ln(y)}{n}\right)^n \\
        & = \ln(2)^n \left(1-\frac{\ln{y}}{n\ln{2}}\right)^n
\end{align}

We combine the $\ln(2)^n$ with the second term, obtaining
\begin{align}
    \ln(2)^n\exp{n-\ln(2)n+\ln(y)}  & = \exp{n\ln(\ln(2))+n(1-\ln(2))+\ln(y)} \\
    & = \exp{n(1-\ln(2)+\ln(\ln(2))}\exp{\ln(y)} \\
    & \approx \exp{-\frac{n}{20}}y
\end{align}

For large $n$, we use $(1-\frac{\ln{y}}{n\ln{2}})^n \approx \exp{-\frac{\ln(y)}{\ln(2)}}$ to find

\begin{align}
    y\exp{-\frac{\ln(y)}{\ln(2)}}\exp{-\frac{n}{20}} & = \exp{\ln(y)-\frac{\ln(y)}{\ln(2)}}\exp{-\frac{n}{20}} \\ 
    & = \exp{\ln(y)\left(1-\frac{1}{\ln(2)}\right)}\exp{-\frac{n}{20}} \\
    & \approx \exp{\ln(y^{-0.44})}\exp{-\frac{n}{20}} \approx \frac{\exp{-\frac{n}{20}}}{\sqrt{y}} 
\end{align}

which concludes the proof.
\end{proof}

\begin{lemma}[Square distance between product distributions]
    Let $p_{\vec a}, q_{\vec b}$ be two instances of the product distributions. Then
    \begin{equation}
    \E{\vec a, \vec b}{\Delta(p_{\vec a}, q_{\vec b})} \in \exp{-\mathcal{O}(n)}
    \end{equation}
\end{lemma}

\begin{proof}
The proof is achieved by explicitly computing mean and variance of this probability distribution. We recall that the probability distributions are symmetric, hence
\begin{equation}
    \E{\vec a, \vec b}{\Delta(p_{\vec a},q_{\vec b})} = 2^n \E{\vec a, \vec b}{\left(p(0) - q(0)\right)^2}, 
\end{equation}
and again by symmetry
\begin{equation}
    \E{\vec a, \vec b}{\left(p(0) - q(0)\right)^2} = 2 \E{\vec a}{p(0)^2} - 2 \E{\vec a}{p(0)}^2. 
\end{equation}
By employing the same tools as in \Cref{le.marginals_product}, both quantities can be easily computed, as
\begin{equation}
    \begin{aligned}
        \E{\vec a}{p(0)} = \int \prod_{i = 1}^n a_i \prod_{i = 1}^n d a_i = \frac{1}{(n - 1)!} & \int_{0}^\infty ds \exp{-2 s} s^{n - 1} = \\
        \frac{1}{(n - 1)!} \frac{1}{2^{n}} & \underbrace{\int_{0}^\infty (2 ds) \exp{-2 s} (2s)^{n - 1}}_{\Gamma(n) = (n-1)!} = \frac{1}{2^n}
    \end{aligned}
\end{equation}

The attentive reader might have seen a simpler way of computing this, since $\E{}{a_i} = 1/2$. We choose to keep it in the form of the gamma function for consistency with the rest of the proof. 

\begin{equation}
    \begin{aligned}
        \E{\vec a}{p(0)^2} = \int \prod_{i = 1}^n a_i^2 \prod_{i = 1}^n d a_i = \frac{1}{(n - 1)!} & \int_{0}^\infty ds \exp{-3 s} s^{n - 1} = \\
        \frac{1}{(n - 1)!} \frac{1}{3^{n}} & \underbrace{\int_{0}^\infty (3 ds) \exp{-3 s} (3s)^{n - 1}}_{\Gamma(n) = (n-1)!} = \frac{1}{3^n}
    \end{aligned}
\end{equation}
Joining all results, we find
\begin{equation}
    \E{\vec a, \vec b}{\Delta(p_{\vec a}, p_{\vec b})} = 2\left(\frac{2^n}{3^n} - \frac{1}{2^n} \right) \leq \exp{-\ln(3/2) n} \in \exp{-\mathcal O(n)}
\end{equation}
\end{proof}

\subsection{Pseudo-Independent Distributions}\label{app:independent-distribution}
Consider the pseudo-independent distributions as defined in \Cref{def.independent_distribution}. Useful symbols in this section will be
\begin{align}
    \mu & = \E{}{Y} \\
    \hat \mu & = \frac{1}{N} \sum_{x} Y_x \\
    \mu_2 & = \E{}{Y^2} \\
    \sigma^2 & = \mu_2 - \mu^2
\end{align}

\begin{lemma}[Difference to average]\label{le.diff_average_Y}
Let $X$ be a N-dimensional random variable as defined in \Cref{def.independent_distribution} and $\{Y_i\}$ its associated i.i.d. random variables.  
Then
\begin{equation}
    \prob{\frac{\vert \mu - \hat\mu \vert}{\mu} \geq k} \leq \frac{\sigma^2}{N k^2 \mu^2} .
\end{equation}
    
\end{lemma}

\begin{proof}   
    Consider the random variable $\hat\mu$. Its variance is given by
    \begin{equation}
        \Var{\hat\mu} = \frac{\sigma^2}{N}, 
    \end{equation}
    because the variance of a sum of weighted independent random variables is given by
    \begin{equation}
        \Var{\sum_j a_j x_j} = \sum_j a_j^2 \Var{x_j}.
    \end{equation}
    
    We recall now Chebyshev's inequality. For any random variable, 
    \begin{equation}
        \prob{\vert x - \E{}{x} \vert \geq t} \leq \frac{\Var{x}}{t^2}. 
    \end{equation}
    This implies that 
    \begin{equation}
        \prob{\vert\hat\mu - \mu \vert \geq t} \leq \frac{\sigma^2}{N t^2}
    \end{equation} 
    The final result is obtained by normalisation of the distance $t \rightarrow k \E{}{Y}$.
\end{proof}

\begin{corollary}\label{cor.diff_average_Y_inv}
    Let $X$ be a N-dimensional random variable as defined in \Cref{def.independent_distribution} with associated i.i.d. random variables $\{Y_i\}$, and let $\mu$ and $\hat{\mu}_N$ as in \Cref{le.diff_average_Y}.
    Then 
    \begin{equation}
        \prob{\mu \left\vert \frac{1}{\mu} - \frac{1}{\hat\mu_N} \right\vert\geq k }\leq \frac{\sigma^2}{N k^2 \mu^2} .
    \end{equation}
\end{corollary}

\begin{proof}
    This corollary follows immediately from \Cref{le.diff_average_Y}. By error propagation
    \begin{equation}
        \mu \left\vert \frac{1}{\mu} - \frac{1}{\hat\mu_N} \right\vert \approx \frac{\vert \mu - \hat\mu\vert}{\mu}\left( 1 + \mathcal O\left( \frac{\vert \mu - \hat\mu\vert}{\mu}\right)\right), 
    \end{equation}
    which yields a first-order equivalence between this result and that of \Cref{le.diff_average_Y}. 
\end{proof}

To address the anticoncentration, we can find the following lemma. 
\begin{lemma}[Anticoncentration of pseudo-independent distributions]
Let $X$ be a $N$-dimensional pseudo-independent random variable as in \Cref{def.independent_distribution}. Then
\begin{equation}
\prob{X_i \geq \frac{\alpha}{N}} \geq \left( 1 - \alpha \left(1 + \frac{1}{k}\right)\right)^2 \left( 1 - \frac{\sigma^2 k^2}{N \mu^2}\right)\frac{\mu^2}{\sigma^2} 
\end{equation}
For constants $\alpha, k$. For this bound to be informative, $\alpha(1 - k^{-1}) \leq 1$.
\end{lemma}

\begin{proof}
The anticoncentration constraint is achieved by bounding
\begin{equation}
    \prob{X_i \geq \frac\alpha N} = \prob{Y_i \geq \alpha \hat \mu}.
\end{equation}
From \Cref{le.diff_average_Y}, and with a slight modification,
\begin{equation}
    \prob{\frac{\left\vert \hat \mu - \mu\right\vert}{\mu} \geq k^{-1}}\leq  \frac{\sigma^2 k^2}{N\mu^2}.
\end{equation}
This implies 
\begin{equation}
    \prob{Y_i \geq \alpha \hat \mu} \geq \prob{Y_i \geq \alpha \mu \left( 1 + \frac{\vert \hat \mu - \mu\vert}{\mu}\right)} \geq \prob{Y_i \geq \alpha \mu \left( 1 + \frac{1}{k}\right)} \left( 1 - \frac{\sigma^2 k^2}{N \mu^2}\right).
\end{equation}
We apply now Paley-Zygmund inequality \cite{paley1932note} to obtain
\begin{equation}
    \prob{Y_i \geq \alpha \mu \left( 1 + \frac{1}{k}\right)} \geq \left( 1 - \alpha\left( 1 + \frac{1}{k}\right)\right)^2 \frac{\mu^2}{\sigma^2}.
\end{equation}
Joining both results
\begin{equation}
    \prob{Y_i \geq \alpha \hat \mu} \geq  \left( 1 - \frac{\sigma^2 k^2}{N \mu^2}\right)\left( 1 - \alpha\left( 1 + \frac{1}{k}\right)\right)^2 \frac{\mu^2}{\sigma^2} .
\end{equation}

To achieve the proposition of interest, we just have to choose adequate values of $\alpha$ and $k$, e. g. $\alpha = 1/3, k = 1/2$, and see that this bound approximates a constant as $N$ increases. 
\end{proof}

It is now straigthforward to state: 
\anticoncentrationindependent*

To prove \Cref{prop.msdpseudo}, we give the following preliminary result.  
\begin{lemma}[Concentration of square distance for pseudo-independent distributions]\label{le.concentration_pseudoindependent}
    Let $p_{\theta_p}(x), q_{\theta_q}(x)$ be two probability distributions following \Cref{def.independent_distribution}. Then, 
    \begin{equation}
        \prob{\Delta(p_{\theta_p}, q_{\theta_q}) \geq \frac{k^2}{N}} \leq 6 \frac{\sigma^2}{\mu^2 k^2} + 6 \frac{\sigma^2 \mu_2}{\mu^4 k^2 N} = 6 \frac{\sigma^2}{\mu^2 k^2} \left( 1 + \frac{\mu_2}{\mu^2 N}\right), 
    \end{equation}
\end{lemma}

\begin{proof}
    We begin with the definition 
    \begin{equation}
        \sum_{x\in \{0, 1\}^n} \left(p_{\theta_p}(x) -  q_{\theta_q}(x)\right)^2 = \sum_{x} \left( \frac{Y_x}{N\hat \mu} - \frac{Y'_x}{N\hat\mu'}\right)^2.
    \end{equation}
    We note that
    \begin{equation}
    \begin{aligned}
        \left( \frac{Y_x}{N\hat \mu} - \frac{Y'_x}{N\hat\mu'}\right)^2 = &  \left( \frac{Y_x}{N\mu} - \frac{Y'_x}{N\mu} + \frac{Y_x}{N\hat \mu} - \frac{Y_x}{N\mu} + \frac{Y'_x}{N\mu} - \frac{Y'_x}{N\hat\mu'} \right)^2 \leq \\
        & 3 \left( \left( \frac{Y_x}{N\mu} - \frac{Y'_x}{N\mu} \right)^2 + \left(  \frac{Y_x}{N\mu} - \frac{Y_x}{N\hat\mu} \right)^2 + \left( \frac{Y'_x}{N\mu} - \frac{Y'_x}{N\hat\mu'} \right)^2\right)
    \end{aligned}
    \end{equation}
    due to convexity of the square function. We can now bound each term separately and calculate its expectation. The first term is simply
    \begin{equation}
            \sum_x \left( \frac{Y_x}{N\mu} - \frac{Y_x'}{N\mu} \right)^2 = \frac{1}{N^2\mu^2} \sum_x (Y_x - Y'_x)^2,
    \end{equation}
    with average is given by
    \begin{equation}
        \E{}{(Y_x - Y_x')^2} = 2 \Var{Y_x} = 2 \sigma^2, 
    \end{equation}
    and hence
    \begin{equation}
        \E{}{\sum_x \left( \frac{Y_x}{N\mu} - \frac{Y_x}{N\mu} \right)^2} = \frac{2 \sigma^2}{N\mu^2}.
    \end{equation}
    Following Markov's inequality, 
    \begin{equation}
        \prob{\sum_x \left( \frac{Y_x}{N\mu} - \frac{Y_x}{N\mu} \right)^2 \geq \frac{k^2}{N}} \leq \frac{2 \sigma^2}{\mu^2 k^2}
    \end{equation}

    For the second and third term, we can just upper bound the differences between $\mu$ and $\hat\mu$ to obtain
    \begin{equation}
        \sum_x \E{}{\left(  \frac{Y_x}{N\mu} - \frac{Y_x}{N\hat\mu}\right)^2} = \frac{1}{\mu^2N^2} \left( \mu^2\left(\frac{1}{\mu} - \frac{1}{\hat\mu}\right)^2 \right)\sum_x \E{}{Y_x^2} = \frac{\mu_2}{N \mu^2}  \left( \mu^2\left(\frac{1}{\mu} - \frac{1}{\hat\mu}\right)^2 \right). 
    \end{equation}
    By \Cref{cor.diff_average_Y_inv}, 
    \begin{equation}
        \prob{\mu^2\left(\frac{1}{\mu} - \frac{1}{\hat\mu}\right)^2 \geq \frac{k^2 \mu^2}{\mu_2}} = \prob{\frac{\mu_2}{N}\left(\frac{1}{\mu} - \frac{1}{\hat\mu}\right)^2 \geq \frac{k^2}{N}} \leq \frac{\sigma^2 \mu_2}{\mu^4 k^2 N}.
    \end{equation}

    Using the union bound, 
    \begin{equation}
        \prob{\Delta(p, q) \geq \frac{k^2}{N}} \leq 6 \frac{\sigma^2}{\mu^2 k^2} + 6 \frac{\sigma^2 \mu_2}{\mu^4 k^2 N} = 6 \frac{\sigma^2}{\mu^2 k^2} \left( 1 + \frac{\mu_2}{\mu^2 N}\right), 
    \end{equation}
    which yields the final result. 
\end{proof}

With the support of the previous result, it is straightforward to state the original proposition: 

\msdpseudo*

\begin{corollary}\label{def.var_uniform}
    Let $p_{\theta_p}(x)$ be an instance of the pseodu-independent probability distributions, and let $u$ be the uniform distribution. Then, 
    \begin{equation}
        \prob{\Delta(p_{\theta_p}, u) \geq \frac{k^2}{N}} \leq \frac{2 \sigma^2}{\mu^2 k^2} \left( 1 + \frac{\sigma^2}{\mu^2 N}\right).
    \end{equation}
\end{corollary}

\begin{proof}
    The proof is analogous to the previous lemma, just substituting $q_{\theta_q}$ with a fixed variable $u(x) = N^{-1}$. The difference $\vert\hat \mu-\mu\vert$ is controlled as before, and we just need to take into account
    \begin{equation}
        \E{}{\left(\frac{Y_x}{\mu} - 1\right)^2} = \frac{\sigma^2}{\mu^2}, 
    \end{equation}
    hence yielding
    \begin{equation}
        \E{}{\sum_x \left( \frac{Y_x}{\mu N} - \frac{1}{N}\right)^2} \leq \frac{\sigma^2}{N \mu^2}\left(1 + \frac{\sigma^2}{N^2 \mu}\right).
    \end{equation}
\end{proof}

\begin{corollary}\label{def.var_Z}
    Let $p_C(x)$ be an instance of the pseudo-independent probability distributions generated by a quantum circuit C. Let $Z$ be any diagonal Pauli string. Then, 
    \begin{equation}
        \Var{\bra 0 C Z C \ket 0} \leq \frac{\sigma^2}{N}\prob{\Delta(p_{\theta_p}, u) \geq \frac{k^2}{N}} \leq \frac{2 \sigma^2}{\mu^2 k^2} \left( 1 + \frac{\mu_2}{\mu^2 N}\right).
    \end{equation}
\end{corollary}

\begin{proof}
    The proof is analogous to the previous lemma. Any Pauli string of this kind splits the $x$ into those associated with $\pm 1$. This implies
    \begin{equation}
        \bra 0 C Z C \ket 0 = \frac{\sum_{x_+} Y_{x_+} - \sum_{x_-} Y_{x_-}}{N \hat \mu}.
    \end{equation}
    As before, 
    \begin{equation}
        \left\vert\bra 0 C Z C \ket 0 - \frac{\sum_{x_+} Y_{x_+} - \sum_{x_-} Y_{x_-}}{N \mu}\right\vert \leq \frac{\sum_{x_+} Y_{x_+} - \sum_{x_-} Y_{x_-}}{N} \left\vert \frac{1}{\hat\mu} - \frac{1}{\mu}\right\vert \leq \frac{\sum_{x_+} Y_{x_+} - \sum_{x_-} Y_{x_-}}{N \mu} \frac{\sigma^2}{N k^2 \mu^2}, 
    \end{equation}
    where the last inequality is given with probability at least $1 - k$. 
    The sum of independent random variables has mean 0, and its variance is given by
    \begin{equation}
        \Var{\frac{\sum_{x_+} Y_{x_+} - \sum_{x_-} Y_{x_-}}{N \mu}} = \frac{\sigma^2}{N \mu^2}
    \end{equation}
    Thus, 
    \begin{equation}
        \Var{\bra 0 C Z C \ket 0} \leq \frac{\sigma^2}{N \mu^2}\left(1 + \frac{\sigma^2}{N \mu^2}\right)
    \end{equation}
\end{proof}

\msdmmd*

\begin{proof}
    We just need to see that 
    \begin{equation}
        \E{}{\mmd{p, q)}} = \sum_{k = 0}^n \frac{(1 + \rho)^k (1 - \rho)^{n-k}}{2^n}\sum_{\substack{S \subseteq [n] \\ \vert S \vert = k}} \frac{\sigma^2}{N\mu^2}.
    \end{equation}
    Identifying that there are $\binom{n}{k}$ elements of length $k$, this is the binomial theorem, and
    \begin{equation}
        \sum_{k = 0}^n \frac{(1 + \rho)^k (1 - \rho)^{n-k}}{2^n} \binom{n}{k} = 1.
    \end{equation}
    The probabilistic statement is a direct consequence of Markov's inequality.
\end{proof}

\subsection{Peaked Distributions}\label{app:peaked-distribution}

\begin{proof}
    The proof follows the same strategy as that of \Cref{le.concentration_pseudoindependent}, just changing the varible $N$ with $K$. For the bounds of $\hat\mu$ and $\mu$, the same results hold. Changes apply in the case of the cross term, that is
    \begin{equation}
        \sum_{i} \E{}{\left(Y_i - Y_i'\right)^2} = 2 K \E{}{Y^2} - 2 J \E{}{Y}^2, 
    \end{equation}
    with $J$ being the number of coincidences of the $K$ out of $N$ elements found in each sample. The number of coincidences $J$ follows a probability distribution 
    \begin{equation}
        \prob{J} = \frac{\binom{N}{K} \binom{K}{J} \binom{N - K}{N - J}}{\binom{N}{K}^2}. 
    \end{equation}
    The numerator counts for (a) the number of combinatorial options in choosing $K$ elements out of $N$ in the first sample, (b) number of possible elements $J$ coinciding from the $K$ original, and $(c)$ number of possibilities to rearrange the remaining terms. The denominator captures the total number of samples. This is known as a hypergeometric distribution \cite{hoeffding1963probability}, with moments
    \begin{align}
        \E{}{J} & = \frac{K^2}{N} \\
        \Var{J} & = \frac{K^2}{N} \frac{N - K}{N} \frac{N-K}{N-1}.
    \end{align}
    In the limit of $K \ll N$, this probability distribution can be approximated by a Poisson distribution with parameter $\lambda = K^2 N^{-1}$. The Chernoff bound used in \Cref{le.marginals_product2} shows that $J = 0$ with probability at least $1 - K^2 N^{-1}$. Therefore, 
    \begin{equation}
        \prob{\Delta(p, q) \geq \frac{k^2}{K}} \leq 6 \frac{\sigma^2}{\mu^2 k^2} + 6 \frac{\sigma^2 \mu_2}{\mu^4 k^2 K} = 6 \frac{\sigma^2}{\mu^2 k^2} \left( 1 + \frac{\mu_2}{\mu^2 N}\right)+\frac{K^2}{N}.
    \end{equation}
\end{proof}

\section{Relation between MMD and Fourier moments}\label{app.mmd-fourier}
In this section we derive the details for the comparisons between $\MMD$ and MSD. As a preliminary, we define the Fourier parity moments of a distribution. The domain of all distributions is $\mathcal{X} = \{0,1\}^n$, with $n$ as the number of qubits. This domains has a direct relation with the group $\mathbb F_2^n$, whose Fourier transform is given as follows. 
\begin{definition}[Fourier characters of a distribution \cite{O’Donnell_2014}]
    For any probability distribution $P(x)$ defined on $\mathcal{X}$, its Fourier decomposition over $\mathbb F_2^n$ is
    \begin{align}
            \hat{P}(S) & =  \sum_{x \in \mathcal{X}} P(x)\, \chi_S(x) \equiv \mathbb{E}_{x\sim P}[\chi_S(x)], 
    \end{align}
    with $\chi_S(x) = (-1)^{\sum_{i \in S} x_i}$.
\end{definition}
Notice that a set $S$ of small cardinality will capture few-body correlations, while larger $\vert S \vert$ captures more global structures. This Fourier description is matched by the Gaussian kernel. 
\begin{definition}[Gaussian kernel]\label{def.hamming_kernel}
    The Gaussian kernel is given by
    \begin{equation}
        K_\varsigma(x,y) = \exp{-\frac{d_H(x,y)}{2\varsigma^2}}, 
    \end{equation}
    with $\varsigma$ as bandwidth and $d_H(x, y)$ the Hamming distance between bitstrings $x$ and $y$.
\end{definition}
Note that the usage of Hamming distance as the internal structure of the kernel allows us to phrase it as a product kernel over individual bits, namely
    \begin{align}
            K_\varsigma(x,y) & = \exp{-\frac{d_H(x,y)}{2\varsigma^2}}  \\
                        & = \exp{-\frac{\sum_{i=1}^{n}|x_i - y_i|}{2\varsigma^2}} \\
                        & = \prod_{i=1}^{n} \exp{-\frac{|x_i - y_i|}{2\varsigma^2}} \\
                        & = \prod_{i=1}^{n}  \exp{-\frac{1}{2\varsigma^2}\mathbbm{1}_{x_i \ne y_i}} \equiv \prod_{i=1}^{n} k(x_i,y_i).
    \end{align}
Under this rewriting, it is clear that the kernel is a positive semidefinite matrix, composed by units of the form 
\begin{equation}
        K^{(1)} = \begin{bmatrix}
            1 & \rho \\ \rho & 1
        \end{bmatrix}, \quad
        \textrm{with } \rho = \exp{-1/2 \varsigma^2}
    \end{equation}

The product structure of the kernel admits a direct interpretation in the form of the Fourier basis previously stated. 
\begin{lemma}[Diagonalizing the kernel through Fourier basis]\label{le.diag_kernel}
    Consider the Gaussian kernel from \Cref{def.hamming_kernel}. It can be rewritten in a diagonal basis as
    \begin{equation}
         K_\varsigma(x, y) = \frac{1}{2^n} \sum_{S\subseteq [n]} (1 - \rho)^{|S|}(1 + \rho)^{n-|S|}\chi_S(x)\chi_S(y), 
    \end{equation}
    with $\chi_S(x)$ being the Fourier characters of a bitstring $x$. 
\end{lemma}
\begin{proof}
    We start by taking the kernel
    \begin{equation}
        K^{(1)} = \begin{bmatrix}
            1 & \rho \\ \rho & 1
        \end{bmatrix},  
    \end{equation}
    which can be diagonalized as
    \begin{equation}
        K_\varsigma= \sum_{k = 0}^{d}\lambda_{k}\phi_k\phi_k^T.
    \end{equation}
    In this case
    \begin{align}
        \lambda_k & = 1+ (-1)^k\rho \\
        \phi_k^T & = \frac{1}{\sqrt 2}\begin{pmatrix}
            1 \\ (-1)^k
        \end{pmatrix}\begin{pmatrix} 1 & (-1)^k\end{pmatrix}
    \end{align}
    Then, a concrete entry for the kernel is recovered as
    \begin{equation}
        K(x, y) = \frac{1 + \rho}{2} + \frac{1 - \rho}{2} \chi_{\{1\}}(x)\chi_{\{1\}}(y)
    \end{equation}
    The kernel is a product kernel, hence
     \begin{align}
            K_\varsigma(x,y) & = \prod_{i=1}^{n} K_\varsigma^{(1)}(x_i,y_i) \\
                & = \prod_{i=1}^{n} \left(\frac{1 + \rho}{2} +\chi_{\{1\}}(x_i)\chi_{\{1\}}(y_i)\frac{1 - \rho}{2}\right) \\
                & = \frac{1}{2^n} \prod_{i=1}^{n} \left(1 + \rho + \chi_{\{1\}}(x)\chi_{\{1\}}(y)(1 - \rho)\right) \\
                & = \frac{1}{2^n}\sum_{S\subseteq [n]}\left( \prod_{i\in S} (1 - \rho)\chi_{\{1\}}(x_i)\chi_{\{1\}}(y_i)\prod_{i\not\in S} (1 + \rho)\right)\\
                & = \frac{1}{2^n} \sum_{S\subseteq [n]} (1 - \rho)^{|S|}(1 + \rho)^{n-|S|}\chi_S(x)\chi_S(y).
    \end{align}
\end{proof}

With this preliminary results, it is then straightforward to write the MMD in the Fourier basis. 

\begin{lemma}[MMD in the Fourier basis]\label{le.mmd_fourier}
    Consider the Gaussian kernel applied to two probability distributions $P, Q$, then
    \begin{equation}
        \mmd{P, Q} = \frac{1}{2^n} \sum_{k = 0}^n (1 + \rho)^{n-k} (1 - \rho)^k \sum_{\substack{S \subseteq [n] \\ \vert S \vert = k}} \left( \hat P(S) - \hat Q(S)\right)^2, 
    \end{equation}
    with $\hat P(s) = \mathbb E_{x \sim P}\left[\chi_S(x)\right]$.
\end{lemma}

\begin{corollary}[Bounding the MMD with mean squared distance]
    For a Gaussian kernel, 
    \begin{equation}
        \mmd{P, Q} \leq \max_S\left( \hat P (S) - \hat Q(S)\right)^2.
    \end{equation}
\end{corollary}

\begin{proof}
    The proof is immediate from \Cref{le.diag_kernel} by using Hölder's inequality, namely 
    \begin{equation}
        \norm{f g}_1 \leq \norm{f}_1 \norm{g}_\infty.
    \end{equation}
    We identify $f$ with the eigenvalues and $g$ with the Fourier coefficients. Then, it is clear that 
    \begin{equation}
        \frac{1}{2^n}\sum_{k = 0} (1 + \rho)^{n-k} (1 - \rho)^k \sum_{\substack{S \subseteq [n] \\ \vert S \vert = k}} 1 = \sum_{k = 0}^n \binom{n}{k}\frac{(1 + \rho)^{n-k} (1 - \rho)^k}{2^n} = 1, 
    \end{equation}
    and the $\infty$-norm is associated with the largest Fourier coefficient. 
\end{proof}

\section{Bounds on the \MMD estimator and statistical test}\label{app:bounds-mmd}

The following bounds are an extension for $m \ne n$ of the bounds derived for the unbiased estimator for the $MMD^2$ and associated statistical test derived in~\cite{JMLR:v13:gretton12a}.  

First, we recall McDiarmid's inequality, which we will use to bound the convergence of the estimator to its true value.

\begin{theorem}[McDiarmid's Inequality\cite{McDiarmid_1989}]
    We consider random variables $X_1, \dots, X_m$ on the domain $\mathcal X$, and a function $f: \mathcal X^m\rightarrow \mathbb R$, where $\forall i\in [m]$ and $\forall x_1,\dots,x_m,x_i'\in\mathcal X$, there exists a $c_i < \infty$ satisfying

    \begin{equation}
        |f(x_1,\dots,x_i,\dots,x_m) - f(x_1,\dots,x_i',\dots, x_m)| \le c_i
    \end{equation}

    Then, $\forall t > 0$

    \begin{equation}
        \Pr_X\big[f(X) - \mathbb E_X(f(X))  > t\big] < \exp{-\frac{2t^2}{\sum_{i=1}^{m}c_i^2}} 
    \end{equation}
\end{theorem}

For a bounded kernel, we can obtain the bounded differences property as follows.

\begin{lemma}[Bounded differences in the \estMMD]\label{lem.bounded-diff}
    Changing a single sample from $X$ or $Y$, will result in changes of at most $\frac{4K}{m}$ or $\frac{4K}{l}$ respectively.
\end{lemma}

\begin{proof}
    A sample in the within-X block of the \estMMD will occur exactly $2(m-1)$ times and from the boundedness of the kernel, each component will change by less than $K$. Factoring in the normalization reveals a change by at most $\frac{2K}{m}$. Analogously, a change of sample in $Y$ will result in a change of at most $\frac{2K}{l}$.

    For the cross term, we observe that a sample from $X$ occurs exactly $l$ times, resulting in a change of at most $\frac{2K}{m}$, and $\frac{2K}{l}$ for a sample from $Y$. Summing up concludes the proof.
\end{proof}

Inserting the bounds into McDiarmid's inequality then reveals the desired bound on the convergence.

\begin{theorem}[Convergence of the \estMMD to its true value]\label{thm.bound-on-statistic}
    Given a bounded kernel $0\le k(\cdot,\cdot)\le K$, the unbiased statistic converges to the true $\MMD$ as

    \begin{equation}
        Pr_{X,Y}\big[|\estMMD - MMD^2| > t \big] < \exp{-\frac{2t^2(m+l)}{8K^2}}
    \end{equation}
\end{theorem}

\begin{proof}
    The proof follows directly from applying McDiarmid's inequality with the bounded differences obtained in \Cref{lem.bounded-diff}
\end{proof}

\begin{corollary}[Statistical test for \estMMD]
    A hypothesis test of level $\alpha$ for the null hypothesis $p=q$ has the acceptance region

    \begin{equation}
        \estMMD \le K\sqrt{\frac{8}{m+l}\log(\alpha^{-1})}
    \end{equation}
\end{corollary}

\begin{proof}
    To obtain an $\alpha$-level test, we consider

    \begin{equation}
        \Pr_{X,Y}\big[|\estMMD| > t] \le \alpha
    \end{equation}

    Inserting the bound from \Cref{thm.bound-on-statistic} and solving for $t$ concludes the proof.
\end{proof}

\section{Computing the TVD for pseudo-independent distributions}

\begin{lemma}[Absolute Difference]\label{le.var_abs_diff}
Let $p, q$ be two distributions from the pseudo-indepdendent family. Then, the 1-norm, defined as 
\begin{equation}
    \norm{p - q}_1 =  \sum_{x} \left\vert p(x) - q(x)\right\vert
\end{equation}
has statistical properties
\begin{align}
    \E{}{\norm{p - q}_1} & = 2 G(Y) + \mathcal O(N^{-1})\\
    \Var{\norm{p - q}_1} & = \frac{4}{N} \left( \frac{\sigma^2}{\mu^2} - G^2(Y)\right) + \mathcal{O}\left(\frac{1}{N^2}\right), 
\end{align}
with $G(Y)$ being the Gini coefficient of the underlying probability distribution. 
\end{lemma}

\begin{proof}
    We consider the random variable from \Cref{def.independent_distribution}. The first step is to compute one element of the distance 
    \begin{equation}
        \sum_x\left\vert \frac{Y_x}{N \hat\mu} - \frac{Y'_x}{N \hat\mu'} \right\vert = \sum_x \left\vert \frac{1}{N\mu} \left( Y_x - Y'_x\right) + \frac{Y_x}{N}\left(\frac{1}{\hat\mu} - \frac{1}{\mu} \right) + \frac{Y'_x}{N}\left(\frac{1}{\hat\mu'} - \frac{1}{\mu} \right) \right\vert, 
    \end{equation}
    and therefore
    \begin{equation}
        \left\vert\sum_x \left\vert \frac{Y_x}{N \hat\mu} - \frac{Y'_x}{N \hat\mu'} \right\vert - \sum_x \frac{1}{N\mu}\left\vert Y_x - Y'_x\right\vert \right\vert \leq \sum_x \frac{Y_x}{N}\left\vert \frac{1}{\mu} - \frac{1}{\hat\mu}\right\vert + \sum_x \frac{Y'_x}{N}\left\vert \frac{1}{\mu} - \frac{1}{\hat\mu'}\right\vert.
    \end{equation}
    Following the steps from \Cref{app:loss-function-concentration}, we can now compute averages to obtain
    \begin{equation}
        \left\vert\sum_x \E{}{\left\vert \frac{Y_x}{N \hat\mu} - \frac{Y'_x}{N \hat\mu'} \right\vert} - \sum_x \frac{1}{N\mu}\E{}{\left\vert Y_x - Y'_x\right\vert} \right\vert \leq 2 \mu \left\vert \frac{1}{\mu} - \frac{1}{\hat\mu}\right\vert.
    \end{equation}

    The approximation to the average can be easily computed as the Gini coefficient \cite{gini} of the independent probability distribution, namely
    \begin{equation}
        G(Y) = \frac{1}{2\mu} \E{}{\vert Y - Y'\vert}. 
    \end{equation}

Hence, the average follows
\begin{equation}
    \sum_x \vert p(x) - q(x)\vert = 2 G(Y) + \mathcal O(N^{-1}).
\end{equation}

The second part of the proof is related to the variance. We can set bounds on each independent term to state
\begin{equation}
    \Var{\vert p(x) - q(x)\vert} = \E{}{\left\vert \frac{Y_x}{\hat\mu N} - \frac{Y'_x}{\hat\mu' N}\right\vert^2} - \E{}{\left\vert \frac{Y_x}{\hat\mu N} - \frac{Y'_x}{\hat\mu' N}\right\vert}^2.
\end{equation}
The first term is analogous to the average of the MSD discussed in \Cref{prop.msdpseudo}. Since the average without the absolute value is zero, then, 
\begin{equation}
     \E{}{\left\vert \frac{Y_x}{\hat\mu N} - \frac{Y'_x}{\hat\mu' N}\right\vert^2} =  \Var{\left\vert \frac{Y_x}{\hat\mu N} - \frac{Y'_x}{\hat\mu' N}\right\vert} =  \Var{\frac{Y_x}{\hat\mu N}} + \Var{\frac{Y'_x}{\hat\mu' N}}.
\end{equation}
This quantity has already been computed in \Cref{def.var_uniform}.

For the second term we can use the triangle bound previously found. Joining all terms together, 
\begin{equation}
    \Var{\sum_x \left\vert p(x) - q(x)\right\vert} \leq \frac{4}{N}\left(\frac{\sigma^2}{\mu^2} - G^2(Y)\right) + \mathcal{O}(N^{-2}).
\end{equation}
\end{proof}

As a reference, the Gini coefficient for the Gamma distribution, which is associated to the Haar-random distribution of states, is given by
\begin{equation}
    G(Y_{\rm gamma}) = \frac{1}{2} \Rightarrow \E{}{\Vert p - q\Vert_1} = 1
\end{equation}

\section{Extension of Numerics}\label{app:extension-numerics}

\begin{figure}
    \includegraphics[width=.5\linewidth]{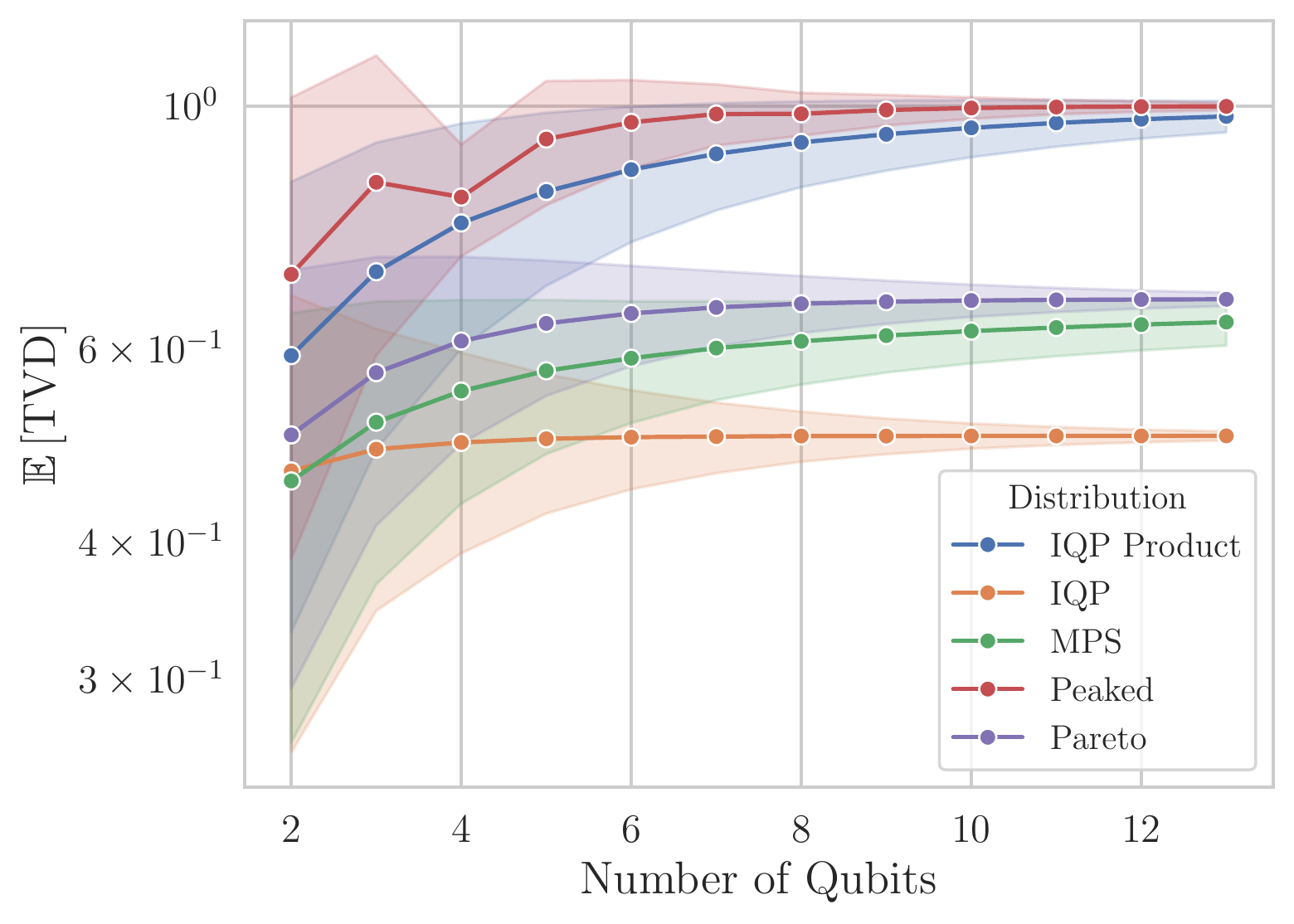}
    \caption{Total Variation Distance}
    \label{fig:numerics-tvd}
\end{figure}

As a comparison to the $\MMD$, we show the TVD in \Cref{fig:numerics-tvd}. While the supremacy distributions also show significantly smaller TVD, it stays approximately the same across system sizes ($\approx 0.5$). Thus, the exponentially vanishing concentration is by no means a property of the distribution families themselves, but the loss function.

\end{document}